\documentclass[sigconf]{acmart}

\usepackage{amsmath,amsfonts,amsthm}
\usepackage{algorithmic}
\usepackage{graphicx}
\usepackage{float}
\usepackage{textcomp}
\usepackage{xcolor}
\usepackage{hyperref}
\usepackage{subfigure}
\usepackage[linesnumbered,ruled,vlined]{algorithm2e}
\usepackage{tabularx}
\usepackage{enumitem}
\usepackage{multirow}

\SetKwRepeat{Do}{do}{while}
\SetKw{Continue}{continue}
\SetKwInput{KwProcedure}{Procedure}

\AtBeginDocument{%
  }

\setcopyright{acmlicensed}
\copyrightyear{2018}
\acmYear{2018}
\acmDOI{XXXXXXX.XXXXXXX}

\acmConference[Conference acronym 'XX]{Make sure to enter the correct
  conference title from your rights confirmation emai}{June 03--05,
  2018}{Woodstock, NY}
\acmISBN{978-1-4503-XXXX-X/18/06}

\begin{document}

\title{Maximum $k$-Plex Search: An Alternated Reduction-and-Bound Method}




\author{Shuohao Gao}
\affiliation{%
  \institution{Harbin Institute of Technology, Shenzhen}
  \country{China}
}
\email{200111201@stu.hit.edu.cn}

\author{Kaiqiang Yu}
\affiliation{%
  \institution{Nanyang Technological University}
  \country{Singapore}
}
\email{kaiqiang002@e.ntu.edu.sg}

\author{Shengxin Liu}
\affiliation{%
  \institution{Harbin Institute of Technology, Shenzhen}
  \country{China}
}
\email{sxliu@hit.edu.cn}

\author{Cheng Long}
\affiliation{%
  \institution{Nanyang Technological University}
  \country{Singapore}
}
\email{c.long@ntu.edu.sg}

\renewcommand{\shortauthors}{Trovato et al.}

\begin{abstract}
  $k$-plexes relax cliques by allowing each vertex to disconnect to at most $k$ vertices. Finding a maximum $k$-plex in a graph is a fundamental operator in graph mining and has been receiving significant attention from various domains. 
  The state-of-the-art algorithms all adopt the branch-reduction-and-bound (BRB) framework where a key step, called \emph{reduction-and-bound} (RB), is used for narrowing down the search space. A common practice of RB in existing works is \texttt{SeqRB}, which \emph{sequentially} conducts the reduction process followed by the bounding process \emph{once} at a branch. However, these algorithms suffer from the efficiency issues.
  In this paper, we propose a new \emph{alternated reduction-and-bound} method \texttt{AltRB} for conducting RB. \texttt{AltRB} first partitions a branch into two parts and then \emph{alternatively} and \emph{iteratively} conducts the reduction process and the bounding process at each part of a branch. With newly-designed reduction rules and bounding methods, \texttt{AltRB} is superior to \texttt{SeqRB} in effectively narrowing down the search space \emph{in both theory and practice}.
  Further, to boost the performance of BRB algorithms, we develop efficient and effective pre-processing methods which reduce the size of the input graph and heuristically compute a large $k$-plex as the lower bound.
  We conduct extensive experiments on 664 real and synthetic graphs. 
  The experimental results show that our proposed algorithm \texttt{kPEX} with \texttt{AltRB} and novel pre-processing techniques runs up to two orders of magnitude faster and solves more instances than state-of-the-art algorithms.
  \end{abstract}

\maketitle

\section{Introduction}\label{sec:into}
The graph model serves as a versatile tool for abstracting numerous real-world data which captures relationships between diverse entities in social networks, biological networks, publication networks, and so on. 
Cohesive subgraph mining is one of the central topics in graph analysis and data mining where the objective is to mine those \emph{dense or cohesive} subgraphs that normally bring valuable insights for analysis~\cite{LRJ+10,CQ18,HLX19,FHQ+20survey,FW+21cohesive}.
For example, cohesive subgraph mining has been used to detect a terrorist cell in social networks~\cite{krebs2002mapping}, to identify protein complexes in biological networks~\cite{zhang2014finding}, and to find a group of research collaborators in publication networks~\cite{guo2022maximal}.

The \emph{clique} is arguably the most well-known cohesive subgraph where every pair of distinct vertices is connected by an edge.
In the literature, the study of efficient algorithms for extracting the maximum clique or enumerating maximal cliques is extensive, e.g., \cite{CP90,PardalosX94a,Tomita17,Cha19,conte2016finding,ELS13,NAUDE201628,TOMITA200628}.
Nevertheless, clique, being a tightly interconnected subgraph, is over-restrictive, which limits its practical usefulness.
To circumvent this issue, relaxations of clique have been proposed and studied in the literature, such as $k$-plex~\cite{seidman1978graph}, $k$-core~\cite{seidman1983network}, quasi-clique~\cite{guo2020scalable,yu2023fast}, and $k$-defective clique~\cite{chang2023efficient,dai2023maximal}.
In particular, $k$-plex relaxes clique by allowing each vertex to disconnect to at most $k$ vertices (including the vertex itself). It is clear that 1-plex corresponds to clique.
The research of cohesive subgraph mining in the context of $k$-plex has recently attracted increasing interests~\cite{XLD+17,wang2023fast,GCY+18,zhou2021improving,JZX+21,CXS22,jiang2023refined,DLQ+22,wang2022listing,chang2024maximum}.

In this paper, we study the \emph{maximum $k$-plex search} problem which aims to search the $k$-plex with the largest number of vertices in the given graph.
It is well-known that the maximum $k$-plex search problem is NP-hard for any fixed $k$~\cite{BBH11}.
Thus, existing studies and ours focus on designing practically efficient algorithms. 

\smallskip

\noindent \underline{\bf Existing algorithms.} The state-of-the-art algorithms all (conceptually) adopt the \emph{branch-reduction-and-bound} (BRB) framework~\cite{GCY+18,zhou2021improving,JZX+21,CXS22,jiang2023refined,wang2023fast,chang2024maximum}. The idea is to recursively solve the problem instance (or branch) by solving the subproblem instances (or sub-branches) produced via a process of \emph{branching}.
A branch denoted by $(S,C)$ corresponds to a problem instance of finding the largest $k$-plex from the subgraph (of the input graph) induced by vertex set $S\cup C$, where the partial solution $S$ corresponds to a $k$-plex and the candidate set $C$ corresponds to the set of vertices used to expand the partial solution. We refer the search space of a branch to the set of possible $k$-plexes in the subgraph induced by $S\cup C$. 
At each branch, a key step, named \emph{reduction-and-bound} (RB), is performed for narrowing down the search space. We note that existing studies all follow a sequential framework, called \texttt{SeqRB}, for implementing the RB step. Specifically, \texttt{SeqRB} \emph{sequentially} conducts two processes \emph{once}: 1) the reduction process shrinks the candidate set $C$ by removing some unpromising vertices that cannot appear in the largest $k$-plex; and 2) the bounding process computes the upper bound of the size of the largest $k$-plex in the branch refined by the first step, which is used for pruning unnecessary branches (i.e., with the upper bound no larger than the largest $k$-plex seen so far). The rationale behind is that with some vertices being removed by the reduction process, the bounding process may obtain a smaller upper bound so as to prune more branches.

Existing studies focus solely on sharpening the reduction rules and upper bound computation methods used in \texttt{SeqRB} while devoting little effort to improving the whole RB framework.
We observe that, in \texttt{SeqRB}, the reduction process benefits the bounding process, but not the other way; and thus they are sequentially conducted only once. 
One interesting question is that: \emph{Can we design a new RB framework where the reduction process and the bounding process can benefit each other}?

We remark that some recent studies~\cite{zhou2021improving,CXS22,jiang2023refined,wang2023fast,chang2024maximum} boost the practical performance of BRB algorithms by devising various pre-processing techniques. These techniques include 1) graph reduction algorithms ~\cite{zhou2021improving,CXS22} for reducing the size of the input graph (among which the best one is \texttt{CTCP}~\cite{CXS22}); and 2) heuristic algorithms~\cite{zhou2021improving,CXS22,chang2024maximum} for computing an initial large $k$-plex used for the above-mentioned reduction algorithms (among which the best ones are \texttt{kPlex-Degen}~\cite{CXS22} and \texttt{EGo-Degen}~\cite{chang2024maximum}).  

\smallskip

\noindent \underline{\bf Our new methods.}
In this paper, we first propose a new framework, called \emph{alternated reduction-and-bound} (\texttt{AltRB}), for conducting the RB step at a branch $(S,C)$. 
\texttt{AltRB} differs from \texttt{SeqRB} mainly in the way of conducting the reduction process and the bounding process.
Specifically, \texttt{AltRB} first partitions a branch into two parts (i.e., $S=S_L\cup S_R$ and $C=C_L\cup C_R$). With newly-designed reduction rules and upper bound computation methods on each part, 
\emph{the bounding process on one part will benefit the reduction process on the other} (note that the reduction process still benefits the bounding process on the same part, which is the same as \texttt{SeqRB}). Thus, \texttt{AltRB} \emph{alternatively} and \emph{iteratively} conducts the reduction process and the bounding process at each part of a branch (e.g., bounding on $S_L\cup C_L$ $\rightarrow$ reduction on $S_R\cup C_R$ $\rightarrow$ bounding on $S_R\cup C_R$ $\rightarrow$ reduction on $S_L\cup C_L$ $\rightarrow$ ...). In this manner, the bounding process and the reduction process could mutually benefit from each other. We show that \texttt{AltRB} is superior to \texttt{SeqRB} in narrowing down the search space in both theory (as will be shown in Equation~(\ref{eq:AltRB_theory})) and practice (as will be shown in Table~\ref{table:comparing-algorithms-k=5}).
We further design efficient pre-processing techniques for boosting the practical performance of BRB algorithms: 1) a new method \texttt{CF-CTCP}, which differs with \texttt{CTCP} in the way of conducting different reductions at each iteration, and 2) a heuristic algorithm \texttt{KPHeuris} that iteratively compute a large initial maximal $k$-plex.

With all the above newly-designed techniques, we develop a new BRB algorithm called \texttt{kPEX}, which runs up to two orders of magnitude faster and solves more instances than state-of-the-art algorithms \texttt{kPlexT}~\cite{chang2024maximum}, \texttt{kPlexS}~\cite{CXS22}, \texttt{KPLEX}~\cite{wang2023fast}, and \texttt{DiseMKP}~\cite{jiang2023refined}.

\smallskip

\noindent \underline{\bf Our contributions.} 
Our main contributions are as follows.
\begin{itemize}[leftmargin=*]
    \item We propose a new BRB algorithm called \texttt{kPEX}, which incorporates the proposed \emph{alternated reduction-and-bound} method \texttt{AltRB} (Section~\ref{sec:existing}). With our newly devised reduction rules and bounding methods, \texttt{AltRB} is superior to \texttt{SeqRB} in narrowing down the search space (Section~\ref{sec:altrb}).
    
    \item We design efficient pre-processing techniques for boosting the performance of BRB algorithms, namely a new method \texttt{CF-CTCP} for reducing the size of the input graph and a heuristic \texttt{KPHeuris} for computing a large initial $k$-plex (Section~\ref{sec:pre-process}).

    \item We conduct extensive experiments on 664 real and synthetic graphs to verify the effectiveness and efficiency of our algorithms. Compared with the state-of-the-art algorithms, our \texttt{kPEX} 1) solves most number of graph instances within the time limit and 2) runs up to two orders of magnitude faster than existing algorithms (Section~\ref{sec:exp}).
\end{itemize}

\section{Preliminaries}\label{sec:preliminary}
Let $G=(V,E)$ be a simple graph with $|V|=n$ vertices and $|E|=m$ edges. 
A vertex $v$ is said to be a neighbor of (or adjacent to) vertex $u$ if there is an edge between $u$ and $v$, i.e., $(u,v) \in E$.
Denote by $N_G(u)=\{v \in V~|~(u,v) \in E \}$ and $d_G(u)=|N_G(u)|$ the neighbor set and the degree of the vertex $u$ in $G$, respectively.
Given a vertex subset $S \subseteq V$, 
we use $G[S]$ to denote the subgraph induced by $S$, i.e., 
$G[S] = (S, \{  (u,v)\in E~|~u,v \in S \})$, and use $N_G(u,S)$ (resp. $\overline{N}_G(u,S)$) to denote sets of neighbors (resp. non-neighbors that include $u$ itself) of $u$ in $G[S]$. We omit the subscript $G$ when the context is clear. Given a graph $g$, we use $V(g)$ and $E(g)$ to denote the sets of vertices and edges in $g$, respectively.

In this paper, we focus on the cohesive subgraph of \emph{$k$-plex}.
\begin{definition}[$k$-plex~\cite{seidman1978graph}]
    Given a positive integer $k$, a graph $g$ is said to be a \emph{$k$-plex} if $d_{g}(u)\geq |V(g)|-k$ for each vertex $u \in V(g)$. 
\end{definition}
Obviously, a 1-plex is a clique where each two vertices are adjacent. 
Note also that $k$-plex has the \emph{hereditary} property, i.e., any induced subgraph of a $k$-plex is also a $k$-plex~\cite{seidman1978graph}.

\smallskip

\noindent \underline{\textbf{Problem statement.}}
Given a graph $G=(V,E)$ and an integer $k\geq 2$, the {\em maximum $k$-plex search problem} aims to find the largest $k$-plex $G[S]$ with $|S|\geq 2k-1$ in $G$. 

Following the previous studies~\cite{CXS22,wang2023fast}, we focus on finding $k$-plexes with at least $2k-1$ vertices for the following considerations. \underline{First}, the value of $k$ is usually small in real applications, e.g., $k \leq 6$ in~\cite{XLD+17,GCY+18,zhou2021improving,JZX+21}. Hence, a $k$-plex with at most $2k-2$ vertices is less informative in practice. \underline{Second}, a $k$-plex with at least $2k-1$ vertices has the diameter of at most 2~\cite{zhou2021improving}, which is more cohesive.

\smallskip

We next introduce some useful concepts used in this paper.

\noindent {\bf $k$-core/$k$-truss.} We review useful cohesive subgraph definitions.
\begin{definition}
    Given a positive integer $k$, a graph $g$ is said to be 
    \begin{itemize}
        \item a \emph{$k$-core} if $d_{g}(u)\geq k$ for each vertex $u \in V(g)$~\cite{seidman1983network};
        \item a \emph{$k$-truss} if each edge $(u,v) \in E(g)$ belongs to at least $k-2$ triangles, i.e., $|N_g(u) \cap N_g(v)| \geq k-2$ for each edge $(u,v) \in E(g)$~\cite{cohen2008trusses}.
    \end{itemize}
\end{definition}
Clearly, a $k$-core $g$ is a $\big(|V(g)|-k\big)$-plex and a $k$-truss $g'$ is a $\big(|V(g')|-k+1\big)$-plex.

\smallskip

\noindent {\bf Degeneracy order.}
The sequence of vertices $v_1, v_2, ..., v_n$ in $G = (V,E)$ is called the {\em degeneracy order} of $G$ 
if $v_i$ has the minimum degree in $G[ \{v_i, v_{i+1}, ..., v_n\} ]$ for each $v_i$ in $V$~\cite{BVZ03m}.
Further, the {\em degeneracy} of $G$, denoted by $\delta(G)$ (or $\delta$ if the context is clear), is defined as the smallest number such that every induced subgraph of $G$ has a vertex of degree at most $\delta(G)$. 
In other words, $G$ does not have an induced subgraph that is a $(\delta+1)$-core.
The degeneracy order and the value of $\delta$ can be obtained by iteratively peeling the vertex with minimum degree in the current induced subgraph with time complexity of $O(m)$~\cite{BVZ03m}. Also, it is known that $\delta \leq \sqrt{n+2m}$~\cite{Cha19}.

\section{The Framework of \texttt{kPEX}}
\label{sec:existing}

Our algorithm, named \texttt{kPEX}, follows the \emph{branch-reduction-and-bound} (BRB) framework which is (conceptually) adopted by existing algorithms~\cite{GCY+18,zhou2021improving,JZX+21,CXS22,jiang2023refined,wang2023fast,chang2024maximum}. The idea is to recursively partition the current problem instance of finding the largest $k$-plex into two subproblem instances via a process of \emph{branching}. Specifically, a problem instance (or branch) is denoted by $(G,S,C)$ (or, simply $(S,C)$ when the context is clear) where the \emph{partial solution} $S$ induces a $k$-plex (i.e., $G[S]$) and the \emph{candidate set} $C$ is a set of vertices that will be used to expand $S$. Solving the branch $(S,C)$ refers to finding the largest $k$-plex $G[H]$ in the branch; \emph{a $k$-plex is in the branch $(S,C)$ if and only if $S\subseteq H\subseteq S\cup C$}. To solve a branch $(S,C)$, it recursively solves two sub-branches formed based on a \emph{branching vertex} $v$ selected from $C$: one branch $(S\cup\{v\},C\setminus\{v\})$ includes $v$ to the partial solution $S$ (which finds the largest $k$-plex containing $v$ in $(S,C)$), and the other $(S,C\setminus\{v\})$ discards $v$ from the candidate set $C$ (which finds the largest $k$-plex excluding $v$ in $(S,C)$). Clearly, solving two formed sub-branches solves branch $(S,C)$, and solving the branch $(\emptyset,V)$ finds the largest $k$-plex in $G$.

\begin{algorithm}[t]
    \caption{Our framework: \texttt{kPEX}}
    \label{alg:basic}
    \KwIn{A graph $G=(V,E)$ and an integer $k$}
    \KwOut{The largest $k$-plex $G[S^*]$}
    \tcc{Stage-I.1: Heuristic\&Preprocessing (Sec.~\ref{sec:pre-process})}
    $S^*\gets $ a large $k$-plex via a heuristic process \texttt{KPHeuris}\; 
    $G\gets$apply reduction method \texttt{CF-CTCP} to reduce $G$\;
    \tcc{Stage-I.2: Divide-and-conquer framework}
    \While{$V(G)\neq \emptyset$}{
        $v\gets$ the vertex with the minimum degree in $G$\;
        $g\gets$ the subgraph of $G$ induced by $N^{\leq 2}(v)$\;
        \tcc{Stage-II:branch-reduction-bound (Sec.~\ref{sec:altrb})}
        \texttt{BRB\_Rec}$(g,\{v\},V(g)\setminus\{v\},k)$\;
        $G\gets$apply reduction method \texttt{CF-CTCP} to reduce $G$\;
    }
    \textbf{return} $G[S^*]$\;
    
    \SetKwBlock{Enum}{Procedure \texttt{BRB\_Rec}$(G,S,C,k)$}{}
    \smallskip
    \Enum{
        $C^{\star}, UB^{\star}\gets$ \texttt{AltRB}$(G,S,C,k)$\;
        \lIf{$UB^{\star}\leq |S^*|$}{\textbf{return}}
        \lIf{$S\!\cup\! C^{\star}$is a $k$-plex}{update $S^*$ by $S\!\cup\! C^{\star}$ and \textbf{return}}
        $v^*\gets$ a branching vertex selected from $C^{\star}$\;
        \texttt{BRB\_Rec}$(G,S\cup\{v^*\},C^{\star}\setminus\{v^*\},k)$\;
        \texttt{BRB\_Rec}$(G,S,C^{\star}\setminus\{v^*\},k)$\;
    }
\end{algorithm}

Our \texttt{kPEX} adopts a similar framework in~\cite{CXS22}, which is summarized in Algorithm~\ref{alg:basic} and involves two stages. 
\textbf{Stage-I} first includes, in Stage-I.1, a heuristic method called \texttt{KPHeuris} for computing a large $k$-plex $G[S^*]$ (maintained globally as the largest $k$-plex seen so far), which will be used to narrow down the search space (Line 1), and a reduction method called \texttt{CF-CTCP} for reducing the input graph $G$ by removing unpromising vertices/edges that will not appear in any $k$-plex larger than $|S^*|$ (Line 2). Besides, \texttt{kPEX} employs a widely-used \emph{divide-and-conquer} strategy in Stage-I.2, which divides the problem of finding the largest $k$-plex in $G$ into several sub-problems (Lines 3-7). Each sub-problem corresponds to a vertex $v$ in $G$ and aims to find the largest $k$-plex that (1) includes vertex $v$ and (2) is in a subgraph of $G$ induced by $v$'s two-hop neighbours $N^{\leq 2}(v)$, i.e., the set of vertices that have distance at most 2 from $v$ (note that a $k$-plex with at least $2k-1$ vertices has the diameter of at most 2~\cite{zhou2021improving} and thus the largest $k$-plex containing $v$ is a subset of $N^{\leq 2}(v)$). Clearly, the largest $k$-plex in $G$ is the largest one among those returned by all sub-problems.
\textbf{Stage-II} corresponds to the recursive process of solving a branch (Lines 9-15). Specifically, \texttt{BRB\_Rec} recursively branches as discussed above (Lines 13-15).
Besides, \texttt{BRB\_Rec} conducts the newly proposed \emph{alternated reduction-and-bound} process (\texttt{AltRB}) on a branch $(S,C)$ for \emph{narrowing down} the search space (Line 10). Specifically, it refines $C$ to $C^{\star}$ by removing some unpromising vertices and computes an upper bound $UB^{\star}$ of (the size of) the largest $k$-plex in $(S,C)$ for terminating the branch.
Finally, we can terminate the branch when (1) $UB^{\star}\leq |S^*|$ since no larger $k$-plex is in the branch and (2) $S\cup C^{\star}$ is a $k$-plex since $G[S\cup C^{\star}]$ is the largest $k$-plex in the branch. 

\smallskip

\noindent\textbf{Novelty.} Our framework differs from the state-of-the-art one~\cite{CXS22} in the following aspects. 
\underline{First}, in Stage-II, \texttt{kPEX} is based on the newly proposed \texttt{AltRB} for narrowing down the search space. Recall that existing methods conduct the reduction-and-bound (RB) process using a sequential method called \texttt{SeqRB} at Line 10 instead.  
We will show that \texttt{AltRB} performs better than \texttt{SeqRB} in Section~\ref{sec:altrb}. Specifically, it refines $C$ to a \emph{smaller} set $C^{\star}$ (i.e., $|C^{\star}|\leq |C|$) and obtains a \emph{tighter} upper bound $UB^{\star}$ (i.e., $UB^{\star}\leq UB$).
\underline{Second}, in Stage-I.1, \texttt{kPEX} employs the novel \texttt{KPHeuris} and \texttt{CF-CTCP} which are more effective and efficient than existing competitors in Section~\ref{sec:pre-process}.

\section{our reduction\&bound method: \texttt{AltRB}}\label{sec:altrb}

\subsection{An Alternated Reduction-and-Bound Method}\label{sec:framework-altrb}
Recall that existing algorithms conduct the reduction-and-bound (RB) step using the sequential method \texttt{SeqRB} on a branch $B=(S,C)$ for narrowing down the search space.
Specifically, \texttt{SeqRB} has two sequential procedures: 1) the \emph{reduction process} refines the candidate set $C$ to $C'$ based on $|S^*|$ (i.e., the lower bound of the branch), i.e., removing from $C$ those vertices that cannot appear in a $k$-plex larger than $|S^*|$; and 2) the \emph{bounding process} obtains the upper bound of the largest $k$-plex in the refined branch $(S,C')$, i.e., the upper bound of the branch denoted by $UB(S,C')$.
In this paper, we propose a new alternated reduction-and-bound method, called \texttt{AltRB}, which is based on a binary partition of a branch $B=(S,C)$ as below.
\begin{equation}
    S=S_L\cup S_R,\ \ C=C_L\cup C_R.
\end{equation}
Let $G[H]$ be a $k$-plex in the branch $B$ such that $G[H]$ is larger than the  largest $k$-plex $G[S^*]$ seen so far, i.e., $|H|\geq |S^*|+1$ (note that other $k$-plexes have the size at most $|S^*|$ and thus can be ignored during the exploration of the branch). 
Based on the above partition, a $k$-plex $G[H]$ in $B$ can be divided into three parts as below.
\begin{equation}
    H=S\cup (C_L\cap H) \cup (C_R\cap H).
\end{equation}

We denote by $LB_L$ and $UB_L$ (resp. $LB_R$ and $UB_R$) the lower and upper bounds of the size of $C_L\cap H$ (resp. $C_R\cap H$), respectively. Formally, we have
\begin{eqnarray}
\label{eq:lower_uppper_def}
    |C_L\cap H| \leq UB_L,\ |C_R\cap H|\leq UB_R.
\end{eqnarray}
Besides, we have the following lemma on the above partition.
\begin{lemma}
\label{lemma:upper_lower}
    Given a branch $(S,C)$ with a partition, we have
    \begin{equation}
        \label{eq:upper_lower}
        |C_L\cap H|\geq (|S^*|+1)-|S|-UB_R,\ |C_R\cap H|\geq (|S^*|+1)-|S|-UB_L.
    \end{equation}
\end{lemma}
\begin{proof}
    This can be easily verified since otherwise if $|C_L\cap H|< (|S^*|+1)-|S|-UB_R$, we have $|H|=|S|+|C_L\cap H|+|C_R\cap H|<|S|+(|S^*|+1)-|S|-UB_R+UB_R = |S^*|+1$, which contradicts with $|H| \geq |S^*|+1$. A similar contradiction can be derived for the other case $|C_R\cap H|< (|S^*|+1)-|S|-UB_L$.
\end{proof}
Based on Lemma~\ref{lemma:upper_lower}, we define $LB_L$ and $LB_R$ as follows.
\begin{equation}\label{eq:lb_L-def}
    (|S^*|+1)-|S|-UB_R\leq LB_L\leq |C_L\cap H|
\end{equation}
\begin{equation}\label{eq:lb_R-def}
    (|S^*|+1)-|S|-UB_L\leq LB_R\leq |C_R\cap H|
\end{equation}

\begin{algorithm}[t]
    \caption{Alternated reduction-and-bound: \texttt{AltRB}}
    \label{alg:altbasic}
    \KwIn{A graph $G=(V,E)$, a branch $(S,C)$ and an integer $k$}
    \KwOut{Refined candidate set $C^{\star}$ and upper bound $UB^{\star}$}
    $S_L,S_R,C_L,C_R\gets$ \texttt{Partition}$(G,S,C,k)$\; 
    $UB_L\gets |C_L|$, $LB_L \gets 0$\;
    \While{$UB_L$ is not equal to \texttt{ComputeUB}$(S_L,C_L)$}{
    $UB_L\gets$\texttt{ComputeUB}$(S_L,C_L)$\;
    $LB_R\gets(|S^*|+1)-|S|-UB_L$; $C_R\gets$ \textbf{RR1}\&\textbf{RR2} on $C_R$\;
    $UB_R\gets$\texttt{ComputeUB}$(S_R,C_R)$\;
    $LB_L\gets(|S^*|+1)-|S|-UB_R$; $C_L\gets$ \textbf{RR1}\&\textbf{RR2} on $C_L$\;
    }
    \textbf{return} $C^{\star}\gets C_L\cup C_R$ and $UB^{\star}\gets |S|+UB_L+UB_R$\;
\end{algorithm}

We note that Lemma~\ref{lemma:upper_lower}
and Equations~(\ref{eq:lb_L-def}) and (\ref{eq:lb_R-def})
indicate the relation between the lower bound of one part and the upper bound of the other, which enables \texttt{AltRB}. We summarize \texttt{AltRB} in Algorithm~\ref{alg:altbasic}, which \emph{iteratively and alternatively} conducts the reduction-and-bound step on the two partitions obtained via \texttt{Partition} (Line 1). Specifically, after initializing $UB_L$ and $LB_L$ in Line 2, \texttt{AltRB} involves the following steps (the details of the two procedures \texttt{Partition} and \texttt{ComputeUB} are provided in Section~\ref{subsec:upper_bound}).
\begin{itemize}[leftmargin=*]
    \item \textbf{Step 1 (Bound on $C_L$).} Compute the upper bound for $C_L$ (i.e., $UB_L$) via a procedure \texttt{ComputeUB} (Line 4). 
    \item \textbf{Step 2 (Reduction on $C_R$).} Update the lower bound for $C_R$ (i.e., $LB_R$) by $(|S^*|+1)-|S|-UB_L$ according to \textbf{Lemma~\ref{lemma:upper_lower}} and then refine $C_R$ based on the updated bounds via reduction rules \textbf{RR1} and \textbf{RR2} (Line 5).
    \item \textbf{Step 3 (Bound on $C_R$).} Compute the upper bound for the refined $C_R$ (i.e., $UB_R$) via a procedure \texttt{ComputeUB} (Line 6). 
    \item \textbf{Step 4 (Reduction on $C_L$).} Update the lower bound for $C_L$ (i.e., $LB_L$) by $(|S^*|+1)-|S|-UB_R$ according to \textbf{Lemma~\ref{lemma:upper_lower}} and then refine $C_L$ based on the updated bounds via reduction rules \textbf{RR1} and \textbf{RR2} (Line 7).
\end{itemize}
Finally, we repeat Steps 1-4 until $UB_L$ remains unchanged (Line 3). We remark that once tighter upper bounds are obtained at Step 1 and Step 3, tighter lower bounds can be derived via Lemma~\ref{lemma:upper_lower} at Step 2 and Step 4 which will be used to boost the performance of \textbf{RR1} and \textbf{RR2}.
Below find the details of reduction rules.

\begin{itemize}[itemindent=0.4cm]
    \item [\textbf{RR1.}]  Given a branch $(S,C)$ with $LB_L$ and $LB_R$, 
    1) for a vertex $v$ in $C_L$, we remove $v$ from $C$ if $|N(v,S\cup C_L)|<LB_L+|S|-k$ or $|N(v,S\cup C_R)|<LB_R+|S|-k+1$; and 2) for a vertex $v$ in $C_R$, we remove $v$ from $C$ if $|N(v,S\cup C_L)|<LB_L+|S|-k+1$ or $|N(v,S\cup C_R)|<LB_R+|S|-k$.
    
    \item [\textbf{RR2.}]  Given a branch $(S,C)$ with $UB_L$ and $UB_R$, 
    1) if $UB_L+UB_R+|S|=|S^*|+1$ and $UB_L=|C_L|$, we move all vertices in $C_L$ from $C$ to $S$ if $G[S\cup C_L]$ is a $k$-plex; otherwise, i.e., it is not a $k$-plex, we terminate the branch $(S,C)$;
    2) if $UB_L+UB_R+|S|=|S^*|+1$ and $UB_R=|C_R|$, we move all vertices in $C_R$ from $C$ to $S$ if $G[S\cup C_R]$ is a $k$-plex; otherwise, i.e., it is not a $k$-plex, we terminate the branch $(S,C)$.
\end{itemize}

\smallskip
\noindent\textbf{Benefits.} Before proving the correctness, we show that \texttt{AltRB} better narrows down the search space than the existing \texttt{SeqRB}.
The rationale behind is based on the following observations. \underline{First}, at \textbf{Step 2} and \textbf{Step 4}, \textbf{RR1}, and \textbf{RR2} (which are based on $UB_L$, $UB_R$, $LB_L$ and $LB_R$) will remove from $C$ more vertices when the lower bounds $LB_L$ and $LB_R$ become larger and/or the upper bounds $UB_L$ and $UB_R$ become smaller; \underline{Second}, at \textbf{Step 1} and \textbf{Step 3}, with some vertices being removed from $C_L$ and $C_R$, smaller upper bound $UB_L$ and $UB_R$ can be derived via \texttt{ComupteUB} (details refer to Section~\ref{subsec:upper_bound}), and larger lower bounds $LB_L$ and $LB_R$ can also be obtained via Lemma~\ref{lemma:upper_lower}; 
\underline{Third}, as \texttt{AltRB} iteratively proceeds, the bounding process and the reduction process will benefit each other (since the former will derive smaller upper bounds and larger lower bounds after the latter while the latter will remove more vertices from $C$ after the former). 
In contrast, \texttt{SeqRB} cannot be conducted iteratively since (1) its reduction rules are only based on $|S^*|$, which will not be changed after \texttt{SeqRB} and (2) thus repeating it multiple times cannot result in either a smaller candidate set $C$ or a smaller upper bound.
We remark that the refined set $C^{\star}$ and the upper bound $UB^{\star}$ obtained by \texttt{AltRB} is potentially \emph{smaller} than those obtained by \texttt{SeqRB} (which will be proved in Section~\ref{subsec:upper_bound}). Thus, with the proposed \texttt{AltRB}, our algorithm \texttt{kPEX} runs up to \emph{two orders of magnitude faster} than the state-of-the-arts, as verified in our experiments.   
 
\smallskip

\noindent\textbf{Correctness.} We then show the correctness of \texttt{AltRB}. Note that \texttt{AltRB} admits an arbitrary partition on $(S,C)$ and any possible procedure for computing $UB_L$ and $UB_R$ that satisfy Equation~(\ref{eq:lower_uppper_def}).

The correctness of \textbf{RR1} can be proved by contradiction. Consider a $k$-plex $G[H]$ in branch $B$ with $|H|\geq |S^*|+1$. Note that if such a $k$-plex does not exist, \textbf{RR1} is obviously correct since all $k$-plexes in branch $B$ are no larger than $|S^*|$ and thus branch $B$ can be terminated. In general, there are two cases. \underline{First}, assume that $G[H]$ contains a vertex $v$ in $C_L$ such that $|N(v,S\cup C_L)|<LB_L+|S|-k$. We get the contradiction by showing that $v$ has more than $k$ non-neighbours in $H$ and thus $G[H]$ is not a $k$-plex since $|N(v,H)|=|N(v,H\cap (S\cup C_L))|+|N(v,H\cap C_R)|\leq (LB_L+|S|-k-1)+|H\cap C_R|\leq (|S|+|H\cap C_R|+|H\cap C_L|)-(k+1)= |H|-(k+1)$. \underline{Second}, assume that $G[H]$ contains a vertex $v$ in $C_L$ such that $|N(v,S\cup C_R)|<LB_R+|S|-k+1$. Similarly, we derive the contradiction by showing that $v$ has more than $k$ non-neighbours in $H$ and thus $G[H]$ is not a $k$-plex since $|N(v,H)|=|N(v,H\cap (S\cup C_R))|+|N(v,H\cap C_L)|\leq (LB_R+|S|-k)+(|H\cap C_L|-1)\leq (|S|+|H\cap C_R|+|H\cap C_L|)-(k+1)= |H|-(k+1)$ (note that $|N(v,H\cap C_L)|\leq |H\cap C_L|-1$ since $v$ is in $C_L$ and is not adjacent to itself). Symmetrically, we can prove the correctness for the reduction rules on $C_R$.

The correctness of \textbf{RR2} is easy to verify. Consider a branch $(S,C)$ with  $UB_L+UB_R+|S|=|S^*|+1$ and $UB_L=|C_L|$, and a $k$-plex $G[H]$ in $(S,C)$ with $|H|\geq |S^*|+1$ (note that if such a $k$-plex does not exist, \textbf{RR2} is obviously correct on this branch). We note that $G[H]$ must contain all vertices in $C_L$, i.e., $C_L\subseteq H$, since otherwise $|H|= |H\cap S|+|H\cap C_L|+|H\cap C_R|\leq |S|+(|C_L|-1)+|H\cap C_R|\leq |S|+UB_L+UB_R-1=|S^*|$. Therefore, $G[S\cup C_L]$ must be a $k$-plex due to the hereditary property; otherwise, such a $k$-plex cannot exist in $(S,C)$ and we can terminate the branch.  

The correctness of \texttt{AltRB} can then be easily verified.

\subsection{Upper Bound Computation and Greedy Partition Strategy}
\label{subsec:upper_bound}
In this part, we first introduce the method \texttt{ComputeUB} used at \textbf{Step 1} and \textbf{Step 3} for obtaining $UB_L$ and $UB_R$ in Section~\ref{sec:framework-altrb}. To boost the performance of \texttt{ComputeUB} as well as the reduction rules on $C_L$ and $C_R$, we then propose a greedy strategy \texttt{Partition} for partitioning $C$ (resp. $S$) into $C_L$ and $C_R$ (resp. $S_L$ and $S_R$). Finally, with all carefully-designed techniques above, we show that the resulted upper bound $UB^{\star}$ will be potentially \emph{smaller} than the existing one $UB$. 

\smallskip

\noindent \underline{\textbf{Upper bound computation.}} We adapt an existing upper bound computation~\cite{jiang2023refined}, which we call \texttt{ComputeUB}, for obtaining $UB_L$ and $UB_R$. Note that it can handle an arbitrary partition on a branch $(S,C)$. Consider \textbf{Step 1} for computing $UB_L$. \texttt{ComputeUB}$(S_L,C_L)$ first iteratively partitions $C_L$ into $(|S_L|+1)$ disjoint subsets. The $i$-th ($1\leq i\leq |S_L|$) subset $\Pi_i(S_L,C_L)$ contains all non-neighbours of a vertex $u_i\in S_L$ in $C_L - \{\Pi_{1}(S_L,C_L),...,\Pi_{i-1}(S_L,C_L)\}$, formally,
\begin{equation}
    \Pi_i(S_L,C_L)=\overline{N}(u_i,C_L^i),\ C_L^i=C_L-\cup_{j=1}^{i-1}\Pi_j(S_L,C_L),
\end{equation}
where $u_i$ is the vertex in $S_L\setminus\{u_1,u_2,...,u_{i-1}\}$ with the largest ratio of $|\overline{N}(u_i,C_L^i)|/(k-|\overline{N}(u_i,S)|)$. Note that the strategy of selecting $u_i$ from $S_L$ has been shown to boost the practical performance of \texttt{ComputeUB} (details refer to~\cite{jiang2023refined}). Besides, we have $\Pi_0(S_L,C_L)=C_L-\{\Pi_{1}(S_L,C_L),...,\Pi_{|S_L|}(S_L,C_L)\}$. Thus, vertices in $\Pi_i(S_L,C_L)$ ($1\leq i\leq |S_L|$) are the non-neighbours of $u_i$ in $C_L$, and vertices in $\Pi_0(S_L,C_L)$ are common neighbours of vertices in $S_L$. The key observation is that for a $k$-plex $G[H]$ in the branch, \emph{$C_L\cap H$ contains at most $\min\{|\Pi_i(S_L,C_L)|, k-|\overline{N}(u_i,S)|\}$ vertices from $\Pi_i(S_L,C_L)$ for $1\leq i \leq |S_L|$} since otherwise $u_i$ (in $H$) will have more than $k$ non-neighbours in $G[H]$ and thus $G[H]$ is not a $k$-plex. Thus, the upper bound $UB_L$ returned by \texttt{ComputeUB}$(S_L,C_L)$ gives as below:
\begin{equation}
    \label{eq:ub_partition}
    |\Pi_0(S_L,C_L)|+\sum_{i=1}^{|S_L|}\min\{|\Pi_i(S_L,C_L)|, k-|\overline{N}(u_i,S)|\}.
\end{equation}
We note that with some vertices being removed from $C_L$ during \texttt{AltRB}, $\Pi_i(S_L,C_L)$ will get smaller and thus a smaller upper bound can be derived. 
Similarly, we can obtain $UB_R$ by \texttt{ComputeUB}$(S_R,C_R)$. 
Besides, we remark that the state-of-the-art upper bound of $k$-plex in the branch $(S,C)$ (used in \texttt{SeqRB}) is $|S|$+\texttt{ComputeUB}$(S,C)$~\cite{jiang2023refined}.

\smallskip

\noindent \underline{\textbf{Greedy partition.}} Consider the upper bound computation at $C_L$, i.e., Equation~(\ref{eq:ub_partition}). We observe that \emph{all vertices in $\Pi_0(S_L,C_L)$ contributes to the upper bound \texttt{ComputeUB}$(S_L,C_L)$} since each of them is adjacent to all vertices in $S_L$ and thus they could appear in a $k$-plex in branch $(S,C)$. The similar observation can be derived on other subsets $\Pi_i(S_L,C_L)$ such that $|\overline{N}(u_i,C_L^i)| \leq k-|\overline{N}(u_i,S)|$ and $1\leq i \leq |S_L|$ (note that there are \emph{fewer} missing edges between $S_L$ and those subsets). Therefore, the adapted upper bound computation performs worse on those subsets.

Motivated by the above observation, we propose to divide $S$ and $C$ into the one ($S_L$ and $C_L$) with \emph{more} missing edges and the other ($S_R$ and $C_R$) with \emph{fewer} missing edges. We summarize the proposed strategy in Algorithm~\ref{alg:partition}. Specifically, we iteratively remove from $S$ to $S_L$ (resp. from $C$ to $C_L$) the vertex $v$ with the greatest value of $|\overline{N}(v,C)|/(k-|\overline{N}(v,S)|)$ (resp. the set of $v$'s non-neighbours in $C$, i.e., $\overline{N}(v,C)$) until the greatest value of $|\overline{N}(v,C)|/(k-|\overline{N}(v,S)|)$ is not greater than 1 or $S$ becomes empty (Lines 2-6). Then, all remaining vertices in $S$ and $C$ will be removed to $S_R$ and $C_R$ (Line 7). We observe that (1) \texttt{ComputeUB}$(S_L,C_L)$ will return a tighter bound since $|\overline{N}(u_i,C_L^i)|>k-|\overline{N}(u_i,S)|$ holds for $1\leq i\leq |S_L|$ and $\Pi_0(S_L,C_L)=\emptyset$, and (2) \texttt{ComputeUB}$(S_R,C_R)$ is always equal to $|C_R|$. 

Consider a branch $(S,C)$ (which has been refined by  \texttt{SeqRB}) with the upper bound $UB=|S|$+\texttt{ComputeUB}$(S,C)$. With the proposed techniques, \texttt{AltRB} will further narrow down the search space of $(S,C)$ by the following observation.
\begin{equation}
\label{eq:AltRB_theory}
    UB^{\star}\leq UB\ \text{and}\ |C^{\star}|\leq |C|.
\end{equation}
We note that $|C^{\star}|\leq |C|$ is obvious since some vertices in $C$ could be removed via \textbf{RR1} and \textbf{RR2}. Besides, $UB^{\star}\leq UB$ holds since (1) \texttt{ComputeUB}$(S,C)=$\texttt{ComputeUB}$(S_L,C_L)+$\texttt{ComputeUB}$(S_R,C_R)$ before \texttt{AltRB} (which can be verified based on the definitions) and (2) as \texttt{AltRB} proceeds, $C_L$ and $C_R$ are refined via \textbf{RR1} and \textbf{RR2}, and thus \texttt{ComputeUB}$(S_L,C_L)$ and \texttt{ComputeUB}$(S_R,C_R)$ get smaller.

\smallskip
\noindent
\textbf{Benefits of greedy partition.}
Compared with a random partition, the greedy partition in Algorithm~\ref{alg:partition} has the following advantageous properties. (1) A tight upper bound of $C_L$ leads to a larger $LB_R$, which enhances the effectiveness of \textbf{RR1}. (2) $UB_R = |C_R|$ is always satisfied, which means that \textbf{RR2} is applicable as long as $UB_L + UB_R + |S| = |S^*| + 1$. In other words, the conditions for \textbf{RR2} are more relaxed. Moreover, computing $UB_R$ as $|C_R|$ is easy to implement and requires less computation.

\begin{algorithm}[t]
    \caption{\texttt{Partition}$(G, S, C, k)$}
    \label{alg:partition}
    \KwIn{Branch $(S,C)$, a graph $G=(V,E)$, and an integer $k$}
    \KwOut{The greedy partition $S_L$, $S_R$, $C_L$ and $C_R$}
    $S_L\gets \emptyset$, $S_R\gets \emptyset$, $C_L\gets \emptyset$, $C_R\gets \emptyset$\;
    \While{$S\neq \emptyset$}{
        $v^*\gets \arg\max_{v\in S}|\overline{N}(v,C)|/(k-|\overline{N}(v,S)|)$ \;
        \lIf{$|\overline{N}(v^*,C)|/(k-|\overline{N}(v^*,S)|) \leq 1$}{\textbf{break}}
        $S_L\gets S_L\cup \{v^*\}$, $C_L\gets C_L\cup \overline{N}(v^*,C)$\;
        $S\gets S\setminus \{v^*\}$, $C\gets C\setminus \overline{N}(v^*,C)$\;
    }
    $S_R\gets S$, $C_R\gets C$\;
    \Return $S_L$, $S_R$, $C_L$ and $C_R$\;
\end{algorithm}

\subsection{Time Complexity Analysis}

\noindent\textbf{Time complexity.}
We analyze the time complexity of \texttt{AltRB} as follows.
(1) \texttt{AltRB} first invokes \texttt{Partition} (Algorithm~\ref{alg:partition}). Specifically, Lines 2-6 of Algorithm~\ref{alg:partition} will be conducted at most $|S|$ times, and each iteration needs to compute $|\overline{N}(v,C)|$ for each $v\in S$, which can be done in $O(|S|\times|C|)$. Thus, the time complexity of \texttt{Partition} is $O(|S|^2|C|)$. 
(2) \texttt{AltRB} then iteratively processes Lines 3-7 of Algorithm~\ref{alg:altbasic}.
We note that \texttt{ComputeUB}$(S,C)$ can be computed in $O(|S|^2|C|)$~\cite{jiang2023refined} (Lines 4 and 6). 
For reductions rules in Lines 5 and 7, \textbf{RR1} iteratively removes the vertex in $C_R$ with minimum $|N(v, S\cup C_L)|$ (or $|N(v,S\cup C_R)|$),
and \textbf{RR2} checks whether $S\cup C_R$ is a $k$-plex. Both rules can be done in $O(|C|\times (|S|+|C|))$. 
(3) We also know that $|S|$ is bounded by $\delta(G)+k$; otherwise, we have a $k$-plex $G[S]$ with $|S|>\delta(G)+k$, which will form a $(\delta(G)+1)$-core and thus contradict the definition of $\delta(G)$; 
$|C|$ is bounded by $\delta(G) d$ since $V(g)$ at Line 5 of Algorithm~\ref{alg:basic} is bounded by $\delta(G) d$~\cite{CDD+18,wang2022listing}.
(4) Let $r$ be the number of iterations of Lines 3-7 of Algorithm~\ref{alg:altbasic}. The number of $r$ is quite small in practice (e.g., $r=1.13$ on average in our experiments) and is bounded by $|C|$ (i.e., $\delta(G) d$) since at least one vertex is removed from $C$ in each round until $C$ becomes empty.

Thus, \texttt{AltRB} (Algorithm~\ref{alg:altbasic}) runs in $O(r \times (|S|^2 |C|+|C|^2)) = O(\delta(G)^3d^3 + k^2\delta(G)^2d^2 )$, where $\delta(G)$ is much smaller than $d$ and $n$ in real graphs, as shown in Table~\ref{table:graph-statistics} ($\delta(G) \leq d < n$ in theory).

\begin{figure}[t]
    \centering
    \includegraphics[width=0.48\textwidth]{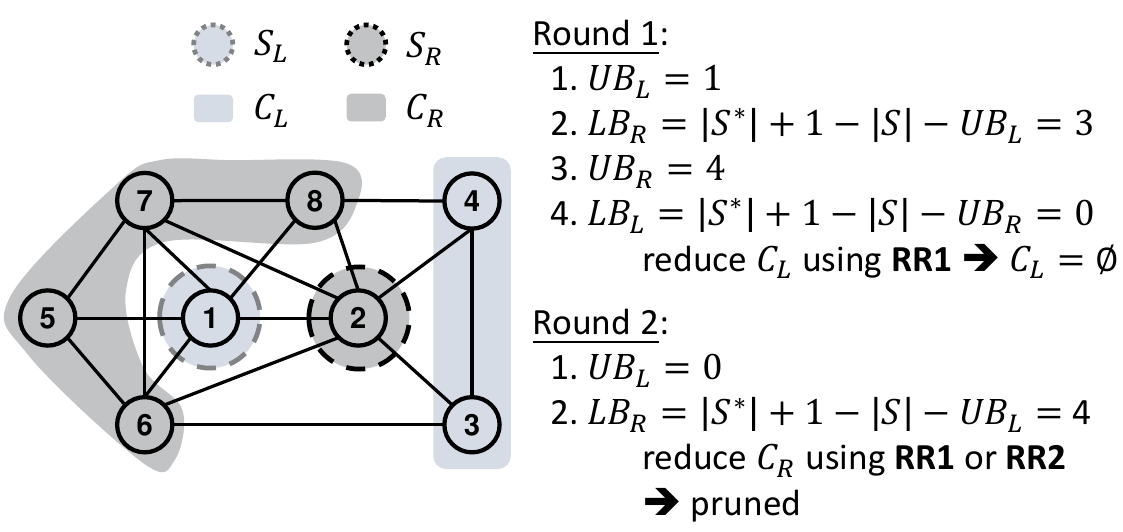}
    \caption{An example of \texttt{AltRB} with $k$=2, $|S^*|$=$5$, $S=\{v_1, v_2\}$, $C=\{v_3, v_4, v_5, v_6, v_7, v_8\}$}
    \label{img-AltRB-example}
\end{figure}

\smallskip
\noindent \textbf{Example.}  To illustrate the proposed \texttt{AltRB}, consider an example in Figure~\ref{img-AltRB-example} with $k$=2, $|S^*|=5$, $S=\{v_1, v_2\}$ and $C=\{v_3, v_4, v_5, v_6, v_7, v_8\}$. 
First, we apply the greedy partition (Algorithm~\ref{alg:partition}) and obtain $S_L = \{v_1\}$, $S_R = \{v_2\}$, $C_L = \{v_3, v_4\}$, and $C_R = \{v_5, v_6, v_7, v_8\}$.
Then, in the first round of \texttt{AltRB} (Lines 4-7 in Algorithm~\ref{alg:altbasic}), we conduct the four steps. (\textbf{Step 1}) Compute the upper bound of $C_L$, i.e., $UB_L = 1$. (\textbf{Step 2}) Update the lower bound of $C_R$ (i.e., $LB_R = 3$) and reduce $C_R$ via \textbf{RR1} and \textbf{RR2} (no vertices are removed). (\textbf{Step 3}) Compute the upper bound of $C_R$, i.e., $UB_R = 4$. (\textbf{Step 4}) Update the lower bound of $C_L$ (i.e., $LB_R = 0$) and reduce $C_L$ via \textbf{RR1} and \textbf{RR2} ($C_L$ is reduced to an empty set).
Next, in the second round with $C_L = \emptyset$, we (1) compute $UB_L = 0$, and (2) update $LB_R = 4$ and reduce $C_R$. If we first apply \textbf{RR1} to $C_R$, 
$v_5, v_6, v_7$ and $v_8$
will be removed, and finally compute the upper bound as $UB^{\star} = |S| + |UB_L| + |UB_R| = 2 + 0 + 0 = 2$, resulting in pruning.
If we first apply \textbf{RR2}, both $UB_L + UB_R +|S| = |S^*| + 1$ and $UB_R = |C_R|$ are satisfied. We then find that $G[S \cup C_R]$ is not a $k$-plex, which means that \textbf{RR2} also leads to pruning.
Actually, the size of maximum 2-plex is 5, indicating that the branch $(S, C)$ cannot find a larger 2-plex, and thus this branch can be pruned by \texttt{AltRB}.
However, without \texttt{AltRB}, the existing method~\cite{jiang2023refined} will compute an upper bound as $UB = UB_L + UB_R + |S| = 1 + 4 + 2 = 7$, which cannot prune the current branch.

\smallskip
\noindent \textbf{Remarks.} We remark that the existing reduction rules proposed in~\cite{CXS22,wang2023fast,chang2024maximum} are all based on $|S^*|$ and thus orthogonal to \texttt{AltRB}. 
We conduct some of these reduction rules to improve practical performance, including (1) additional reduction on subgraph $g$ (Lemma 3.2 in~\cite{CXS22} and Reduction 2 in~\cite{wang2023fast}) in Line 5 of Algorithm~\ref{alg:basic}, and (2) reduction on $C$ before \texttt{AltRB} (\textbf{RR4} in~\cite{CXS22} and Algorithm 3 in~\cite{chang2024maximum}).
Besides, \texttt{AltRB} is also orthogonal to the branching rules for selecting the branching vertex and forming the sub-branches.

\section{Efficient Pre-processing techniques}
\label{sec:pre-process}
In this section, we develop some efficient pre-processing techniques for further boosting the performance of BRB algorithms, namely, \texttt{CF-CTCP} for reducing the size of the input graph in Section~\ref{subsec:CF-CTCP} and \texttt{KPHeuris} for heuristically computing a large $k$-plex in Section~\ref{subsec:kpheuris}.

\subsection{Faster Core-Truss Co-Pruning: \texttt{CF-CTCP}}
\label{subsec:CF-CTCP}

Let $lb$ be the lower bound of the size of the largest $k$-plex (which corresponds to the size of the largest $k$-plex $G[S^*]$ seen so far). We also let $\Delta(u,v)$ be the set of common neighbors of $u$ and $v$ in $G$, i.e., $\Delta(u,v)= N_G(u) \cap N_G(v)$. The idea of refining the input graph $G$ is to \emph{remove from $G$ those vertices and edges that cannot appear in any $k$-plex larger than $lb$ as many as we can}. Existing methods~\cite{zhou2021improving,CXS22,jiang2023refined} are all based on the following lemmas and differ in the implementations (the details of proof is omitted for the ease of presentation).

\begin{lemma}\label{lemma:core-pruning}
    (Core Pruning~\cite{GCY+18}) For each vertex $u \in V(G)$, $u$ cannot appear in a $k$-plex of size $lb+1$ if $d_G(u) \leq lb-k$.
\end{lemma}
\begin{lemma}\label{lemma:truss-pruning}
    (Truss Pruning~\cite{zhou2021improving}) For each edge $(u,v) \in E(G)$, $(u,v)$ cannot appear in a $k$-plex of size $lb+1$ if $\delta_G(u,v) \leq lb-2k$ where $\delta_G(u,v)$ is the number of common neighbors of $u$ and $v$, i.e., $\delta_G(u,v) = |\Delta(u,v)|$.
\end{lemma}

Note that the time complexities of core pruning and truss pruning are $O(m)$~\cite{BVZ03m} and $O(m\times \delta(G))$~\cite{WC12}, respectively.
The above lemmas (namely, core pruning and truss pruning) indicates those unpromising vertices and edges can be removed from $G$. In particular, with some vertices or edges being removed from $G$, the remaining vertices $u$ and edges $(u,v)$ have $d_G(u)$ and $\delta_G(u,v)$ decreases, respectively; and then more vertices and edges can be removed. Therefore, the core pruning (resp. the truss pruning) can be conducted in an iterative way, i.e., iteratively removing unpromising vertices (resp. edges) and updating $d_G(\cdot)$ (resp. $\delta_G(\cdot,\cdot)$) for the remaining until no vertex or edge can be removed. 
We remark that the state-of-the-art method called the core-truss co-pruning (\texttt{CTCP}~\cite{CXS22}) iteratively conducts the truss pruning and then the core pruning in multiple rounds until the graph remains unchanged. However, we observe that \texttt{CTCP} is still inefficient due to potential redundant computations. This is because (1) \texttt{CTCP} performs the truss pruning and the core pruning \emph{separately} at each round (i.e., first remove a set of edges via the truss pruning and then remove one unpromising vertex via core pruning), (2) the truss pruning has the time complexity of $O(m \times \delta(G))$ lager than $O(m)$ for the core pruning, and (3) we note that during the truss pruning, some vertices can be removed via the more efficient core pruning while the truss pruning will iteratively check all their incident edges and then remove some of them (which is very costly).

To improve the practical efficiency of \texttt{CTCP}, we propose a new algorithm called the \emph{core-pruning-first core-truss co-pruning} (or \texttt{CF-CTCP}), which differs from \texttt{CTCP} in the way of conducting pruning at each round. 
Specifically, at each round, it first removes \emph{all} unpromising vertices and then removes \emph{one} unpromising edge (recall that \texttt{CTCP} first removes \emph{all} unpromising edges and then \emph{one} unpromising vertex). The benefit is that unpromising vertices can be removed immediately via efficient core pruning.
Note that our \texttt{CF-CTCP} has the same output but requires less computation compared to \texttt{CTCP}. Given the integer $k$ and the lower bound $lb$, both \texttt{CF-CTCP} and \texttt{CTCP} reduce the input graph $G$ to the maximal subgraph that is $(lb+1-k)$-core and $(lb+3-2k)$-truss.
The difference between \texttt{CF-CTCP} and \texttt{CTCP} is illustrated in Figure~\ref{img-rationale-CTCP}.

\begin{figure}[h]
    \centering
    \subfigure[Rationale of \texttt{CTCP}]{
        \includegraphics[width=0.205\textwidth]{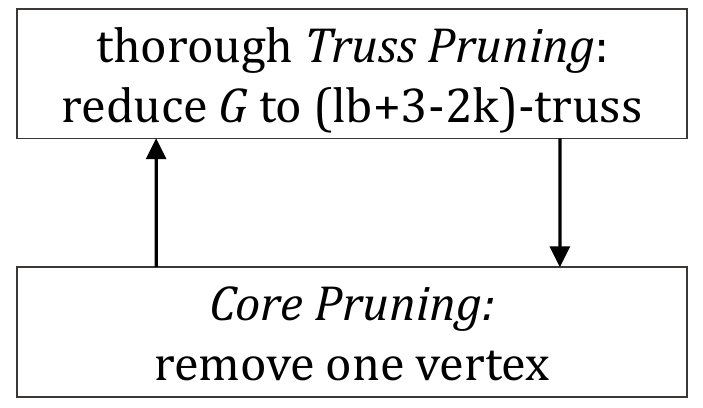}
    }
    \subfigure[Rationale of \texttt{CF-CTCP}]{
        \includegraphics[width=0.205\textwidth]{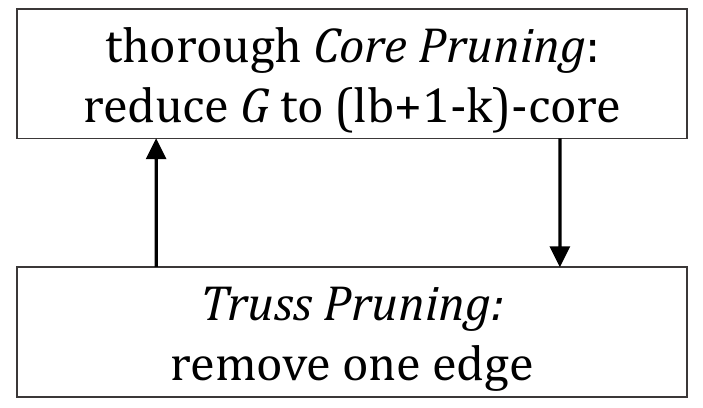}
    }
    \caption{Comparing \texttt{CTCP} and \texttt{CF-CTCP}}
    \label{img-rationale-CTCP}
\end{figure}

The main idea of \texttt{CF-CTCP} is to conduct core pruning thoroughly as follows: 1) if we identify an edge that can be removed, we will immediately remove this edge, even if we have not yet finished computing $\Delta(\cdot,\cdot)$ (i.e., all triangles for each edge); 2) after removing an edge $(u,v)$, we will check whether $u$ or $v$ can be reduced by core pruning. 
Note that after removing an edge $(u,v)$, we postpone the action of updating $\Delta(u, \cdot)$ and $\Delta(v, \cdot)$ since it is time-consuming and there may lead to redundant computations. For example, if both vertices $u$ and $v$ will be removed by core pruning later, updating $\Delta(u, \cdot)$ and $\Delta(v, \cdot)$ is not necessary.

Our proposed \texttt{CF-CTCP} is shown in Algorithm~\ref{algorithm:CF-CTCP}. 
The input of \texttt{CF-CTCP} includes: 1) a set of vertices $Q_v$, which stores the vertices that need to be removed; 2) two integers $\tau_v = lb-k$ and $\tau_e = lb-2k$ that serve as thresholds for the numbers of degrees and triangles for pruning, respectively; 3) a boolean value $lb\_changed$ which is $true$ if a larger $k$-plex is found in \texttt{kPEX} and \texttt{KPHeuris} (Algorithm~\ref{alg:basic} and Algorithm~\ref{algorithm:KPHeuris}). 
We note that both \texttt{kPEX} and \texttt{KPHeuris} (Algorithms~\ref{alg:basic} and~\ref{algorithm:KPHeuris}) invoke \texttt{CF-CTCP} multiple times. 
For example, \texttt{KPHeuris} invokes \texttt{CF-CTCP} by calling \texttt{CF-CTCP}$(G, \emptyset, lb-k, lb-2k, true)$ when it finds a larger heuristic $k$-plex of size $lb$.

We then describe the details of \texttt{CF-CTCP} in steps.
\underline{First}, we design a procedure called \texttt{RemoveEdge} (Lines 21-24) to remove one unpromising edge in Line 21 and all current unpromising vertices in Line 22 via core and truss pruning. The set of removed edges to be considered (due to Lines 21 and 22) is pushed into $Q_e$, which will be used to update $\Delta(\cdot, \cdot)$ later. 
\underline{Second}, Lines 5-6 initialize the sets of common neighbours $\Delta(\cdot, \cdot)$ if \texttt{CF-CTCP} is invoked for the first time. Whenever we find an edge $(u,v)$ that can be reduced, we invoke the procedure \texttt{RemoveEdge} to remove $(u,v)$ immediately in Line 8.
\underline{Third}, we postpone the action of updating $\Delta(\cdot, \cdot)$ to Lines 9-20.
Lines 11-20 consider the effect of each removed edge $(u,v)$ by traversing all the triangles that $(u,v)$ participates in. Specifically, Lines 11-15 traverse each edge $(u,w) \in E$ satisfying $v \in \Delta(u,w)$, i.e., $u,v,w$ can form a triangle, then we update $\Delta(u,w)$ and check whether edge $(u,w)$ can be reduced. Lines 16-20 consider the edges connected to $v$, which is similar to Lines 11-15. Note that in Lines 15 and 20, if we find an edge that can be reduced, we invoke the procedure \texttt{RemoveEdge} to remove the edge immediately.

\begin{algorithm}[t]
    \caption{\texttt{CF-CTCP}$(G=(V,E), Q_v, \tau_v, \tau_e, lb\_changed )$}
    \label{algorithm:CF-CTCP}
   \KwIn{A graph $G=(V,E)$, the set of vertices to be removed $Q_v$, two integral thresholds $\tau_v$ and $\tau_e$, a boolean value $lb\_changed$}
    \KwOut{The reduced graph which is the maximal subgraph in $G$ that is both a $(\tau_v +1)$-core and a $(\tau_e+3)$-truss}
    Remove the vertices in $Q_v$ from $G$ and reduce $G$ to the maximal $(\tau_v+1)$-core by the core pruning\;
     Initialize the set of removed edges to be considered $Q_e \leftarrow \{$\emph{edges removed at Line 1}$\}$\;
    \If{$lb\_changed$}{
        \For{\emph{\bf each} $(u,v) \in E$}{
                \If{\texttt{CF-CTCP} is invoked for the first time}{
                    $\Delta(u,v) \leftarrow  N_{G}(u) \cap N_{G}(v) $\;
                }
                \If{$|\Delta(u,v)| \leq \tau_e$}{
                    $Q_e \leftarrow Q_e \cup$ \texttt{RemoveEdge}$(G, (u,v), \tau_v)$\;
                }
        }
    }
    \While{$Q_e \neq \emptyset$ }{
        $(u,v)\leftarrow $ pop an edge from $Q_e$\;
        \If{$u \in V$}{
            \For{\emph{\textbf{each}} $w \in N_{G}(u)$ satisfying $v \in \Delta(u,w)$}{
                Remove $v$ from $\Delta(u,w)$\;
                \If{$|\Delta(u,w)| \leq \tau_e$}{
                    $Q_e \leftarrow Q_e \cup$ \texttt{RemoveEdge}$(G, (u,w), \tau_v)$\;
                }
            }
        }
        \If{$v \in V$}{
            \For{\emph{\textbf{each}} $w \in N_{G}(v) $ satisfying $u \in \Delta(v,w)$}{
                Remove $u$ from $\Delta(v,w)$\;
                \If{$|\Delta(v,w)| \leq \tau_e$}{
                    $Q_e \leftarrow Q_e \cup$ \texttt{RemoveEdge}$(G, (v,w), \tau_v)$\;
                }
            }
        }
    }

    \BlankLine

    \KwProcedure{\texttt{RemoveEdge}$(G, (u, v), \tau_v)$}
    \KwOut{The set of removed edges to be considered $Q_e$}
    Remove the \emph{unpromising edge} $(u,v)$ from $G$\; 
    Reduce $G$ to the maximal $(\tau_v+1)$-core by the core pruning\;
    Initialize the set of removed edges to be considered $Q_e \leftarrow \{$\emph{edges removed at Lines 21-22}$\}$\;
    \Return $Q_e$\;
\end{algorithm}

\smallskip

\noindent
\textbf{Time complexity.}
We analyze the time complexity of \texttt{CF-CTCP} (Algorithm~\ref{algorithm:CF-CTCP}), including all invocations in \texttt{kPEX}, in the following.

\begin{lemma}\label{lemma:time-complexity-CF-CTCP}
    The total time complexity of all invocations in \texttt{kPEX} (Algorithm~\ref{alg:basic} which includes invocations in the heuristic process \texttt{KPHeuris} in Algorithm~\ref{algorithm:KPHeuris})) to \texttt{CF-CTCP} (Algorithm~\ref{algorithm:CF-CTCP}) is $O(m \times \delta(G))$.
\end{lemma}

The omitted proof, along with an implementation of \texttt{CF-CTCP} with $O(m)$ memory usage, is provided in the appendix.

\smallskip

\noindent
\textbf{Remarks.} \textbf{First}, the time complexity of \texttt{CTCP} is $O(m \times \delta(G) +m\times k)$=$O(m \delta(G))$, requiring that $k$ is a small constant. However, $k$ is up to $n$ in theory and the time complexity of our \texttt{CF-CTCP} is always $O(m \delta(G) )$ for all possible values of $k$.
\textbf{Second}, we do not consider the update of $\Delta(\cdot, \cdot)$ when removing a vertex because removing a vertex is equivalent to first removing all the edges connected to this vertex and then removing this isolated vertex. Therefore, we only consider the removed edges for updating $\Delta(\cdot, \cdot)$.
\textbf{Third}, the acceleration of \texttt{CF-CTCP} can be attributed to two main factors: 1) we do not need to compute the numbers of triangles for the edges that can be removed by core pruning; 2) for an edge $(u,v)$ to be removed such that both $u$ and $v$ are already removed by core pruning, we do not need to traverse related triangles to update $\Delta(\cdot, \cdot)$. Note that if we cannot remove any vertex or edge, the time consumption of \texttt{CF-CTCP} will be the same as \texttt{CTCP} in theory, which is due to the fact that both of them need to compute $\Delta(\cdot, \cdot)$ in $O(m\times \delta(G))$.

\begin{figure}[]
    \centering
    \includegraphics[width=0.45\textwidth]{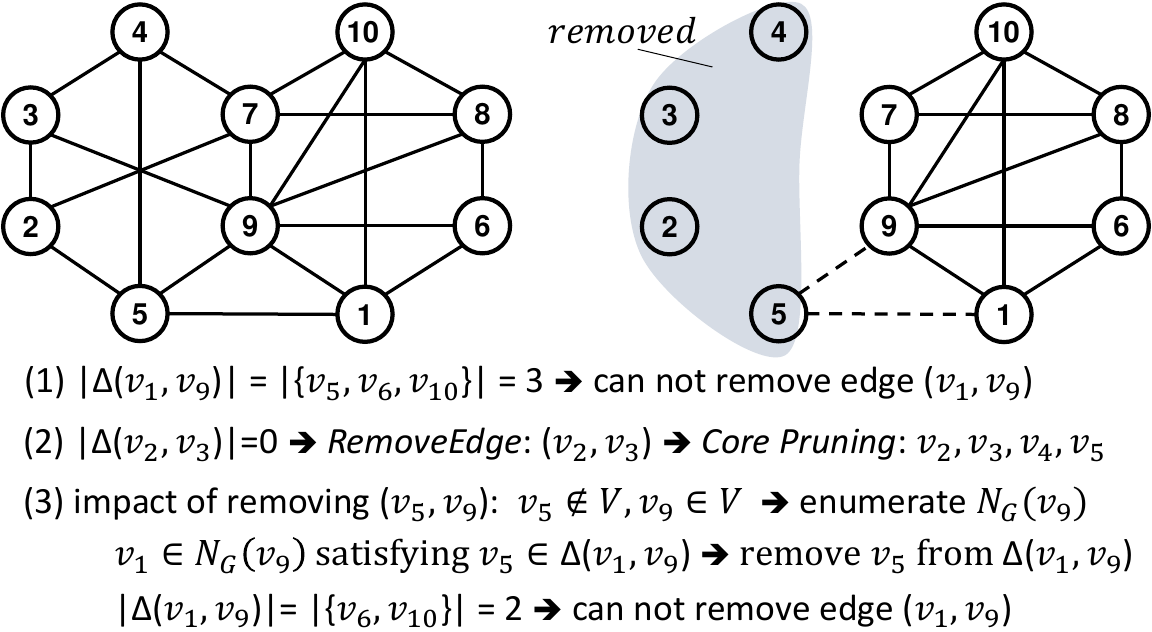}
    \caption{An example for \texttt{CF-CTCP} assuming $lb=4$ and $k=2$}
    \label{img-example-CTCP}
\end{figure}

\smallskip 

\noindent
\textbf{An example of \texttt{CF-CTCP}.} Consider the example of \texttt{CF-CTCP} (Algorithm~\ref{algorithm:CF-CTCP}) in Figure~\ref{img-example-CTCP}, assuming $lb=4$ and $k=2$. According to Lemma~\ref{lemma:core-pruning} and Lemma~\ref{lemma:truss-pruning}, we need to reduce $G$ to the maximal subgraph that is both a $3$-core and a $3$-truss, i.e., we will remove a vertex $u$ if $d_G(u)<3$ and an edge $(u,v)$ if $|\Delta(u,v)| < 1$. 
\underline{First}, we enumerate each edge $(u,v)$ and compute the common neighbors of $u$ and $v$ (Lines 4-6). For those edges connected to $v_1$, we cannot remove them since they have enough common neighbors, e.g., there are 3 common neighbors of $v_1$ and $v_9$. 
However, when we consider edges connected to $v_2$, we find that edge $(v_2, v_3)$ can be removed since $|\Delta(v_2, v_3)|=0$. We then immediately remove edge $(v_2, v_3)$ and conduct core pruning, which removes vertices $v_2, v_3, v_4$, and $v_5$ (Lines 21-24). After this process, we continue to compute common neighbors for the remaining edges, but none of these edges can be removed. 
\underline{Second}, we begin to consider those removed edges in $Q_e$. We focus on the edges $(v_5, v_9)$ and $(v_1, v_5)$ since the other removed edges cannot form a triangle with the remaining edges in $G$. 
For the removal of the edge $(v_5, v_9)$, according to Lines 11-20, we update $\Delta(v_1, v_9)$, and the triangle $(v_1, v_5, v_9)$ is 
destroyed.
Then, for the edge $(v_1, v_5)$, since the triangle $(v_1, v_5, v_9)$ no longer exists after removing the edge $(v_5, v_9)$, the edge $(v_1, v_5)$ cannot constitute any triangle with other vertices. Thus, the procedure of \texttt{CF-CTCP} terminates.
Finally, we reduce $G$ to $G[\{v_1, v_6, v_7, v_8, v_9, v_{10}\}]$ where $G[\{v_1, v_6, v_8, v_9, v_{10}\}]$ is a $2$-plex of size $5$.

\subsection{Compute a large $k$-plex: \texttt{KPHeuris}}
\label{subsec:kpheuris}

We introduce a heuristic method \texttt{KPHeuris} for computing a large initial $k$-plex. Note that such an initial $k$-plex offers a lower bound $lb$, which helps to narrow the search space; and the larger the lower bound is, the more search space can be refined. Therefore, \texttt{KPHeuris} is designed for obtaining a large $k$-plex \emph{efficiently and effectively}.

\begin{algorithm}[tbhp]
    \caption{\texttt{KPHeuris}$(G, k)$}
    \label{algorithm:KPHeuris}
   \KwIn{A graph $G=(V,E)$ and an integer $k>1$}
    \KwOut{The vertex set $S$ of a heuristic initial $k$-plex $G[S]$}
    
    $S \leftarrow $ \texttt{Degen}$(G, k)$, $lb \leftarrow |S|$\;
    Apply \texttt{CF-CTCP} for refining $G$ based on $lb$\;
    \For{\emph{\textbf{each}} $v_i \in V(G)$}{
        $g \leftarrow$ $G[\{v_i, v_{i+1}, ..., v_n\} \cap N^{\leq 2}(v_i)]$;
        $S' \leftarrow$ \texttt{Degen}$(g, k)$\;
        \If{$|S'| > |S|$}{
            $S \leftarrow S'$, $lb \leftarrow |S|$\;
            Apply \texttt{CF-CTCP} for refining $G$ based on $lb$\;
        }
    }
    \Return $S$\;

    \BlankLine

    \KwProcedure{\texttt{Degen}$(G, k)$}
    \KwOut{The vertex set $S$ of a heuristic maximal $k$-plex in $G$}
    $v_1, v_2, ..., v_n\gets$ the degeneracy order of vertices in $V(G)$\;
    $S\gets \emptyset$\; 
    \For{$i=n$ to $1$}{
        \lIf{$G[S \cup \{v_i\}]$ is a $k$-plex}{
            $S \leftarrow S \cup \{v_i\}$
        }
    }
    \Return $S$\;
\end{algorithm}

We summarize \texttt{KPHeuris} in Algorithm~\ref{algorithm:KPHeuris}, which  relies on a sub-procedure (called \texttt{Degen}) for computing a large $k$-plex. Specifically, \texttt{Degen} iteratively includes to an empty set $S$ a vertex in a graph $G$ based on the degeneracy ordering while retaining the $k$-plex property of $G[S]$ until we cannot continue (Lines 9-12). To compute a larger $k$-plexes, \texttt{KPHeuris} further generate $n$ subgraphs from $G$, each of which corresponds to a vertex in $G$ (Lines 3-4); it then invokes \texttt{Degen} on each of them to obtain a $k$-plex (Line 4); it finally returns the largest one among $n+1$ found $k$-plexes. Note that the subgraph related to $v_i$ is the subgraph induced by $\{v_i, v_{i+1}, ..., v_n\} \cap N^{\leq 2}(v_i)$ where $N^{\leq 2}(u)$ denotes $u$'s neighbors and $u$'s neighbors' neighbors, and the rationale is that it can make the subgraph smaller and denser, which tends to find a larger $k$-plex easier.  The time complexity of \texttt{Degen} is $O(m)$, and we will invoke it at most $n+1$ times, thus the total time complexity of computing heuristic solutions in Algorithm~\ref{algorithm:KPHeuris} is $O(nm)$.  
We remark that the total time complexity of all invocations of \texttt{CF-CTCP} is $O(m\delta(G))$ because we invoke \texttt{CF-CTCP} in \texttt{KPHeuris} only when we find a larger $k$-plex, i.e., $lb\_changed=true$, as shown in Lemma~\ref{lemma:time-complexity-CF-CTCP}. Thus, the time complexity of \texttt{KPHeuris} is $O(m\delta(G) + nm) = O(nm)$.

\smallskip
\noindent\textbf{Compared with existing heuristic methods.}
There are two state-of-the-art heuristic methods: \texttt{kPlex-Degen} (\cite{CXS22}) and \texttt{EGo-Degen} (\cite{chang2024maximum}).
\texttt{kPlex-Degen} computes a large $k$-plex by iteratively removing a vertex from the input graph $G$ based on a certain ordering until the remaining graph becomes a $k$-plex. \texttt{KPHeuris} differs from \texttt{kPlex-Degen} in two aspects. First, \texttt{Degen} computes a large $k$-plex by iteratively including a vertex, which is more efficient since the size of the largest $k$-plex is usually much smaller than the size of the input graph (especially for real-world graphs) and can always return a maximal $k$-plex, while \texttt{kPlex-Degen} cannot. Second, \texttt{KPHeuris} further explores $n$ subgraphs instead of the input graph $G$, which tends to obtain a larger $k$-plex as empirically verified in our experiments.
\texttt{EGo-Degen} extracts a subgraph $g_v$ for each vertex $v$ and invokes \texttt{kPlex-Degen} to compute a $k$-plex in $g_v$. Then, \texttt{EGo-Degen} selects the largest $k$-plex among those computed on $n$ subgraphs as the initial heuristic $k$-plex. \texttt{KPHeuris} differs from \texttt{EGo-Degen} in three aspects. 
First, the method of subgraph extraction is different. For a vertex $v \in V(G)$, \texttt{EGo-Degen} extracts $g_v = G[\{v_i, v_{i+1}, ..., v_n\} \cap N_G(v_i)]$, while our \texttt{KPHeuris} generates a subgraph $g_v' = G[\{v_i, v_{i+1}, ..., v_n\} \cap N^{\leq 2}(v_i)]$. It is easy to verify that $g_v \subseteq g_v'$ due to $N_G(v_i) \subseteq N^{\leq 2}(v_i)$. Additionally, a larger subgraph tends to contain a larger $k$-plex, as verified in Table~\ref{table:compare-prepro-k=5}.
Second, \texttt{EGo-Degen} computes $k$-plexes by invoking \texttt{kPlex-Degen}, which implies that it may find a non-maximal $k$-plex as mentioned above.
Third, once a larger $k$-plex is found, \texttt{KPHeuris} updates $lb$ and removes unpromising vertices/edges immediately, while \texttt{EGo-Degen} does not reduce the graph until $n$ heuristic $k$-plexes are computed.

\section{Experimental Studies}
\label{sec:exp}

We test the efficiency and effectiveness of our algorithm \texttt{kPEX} by comparing with the state-of-the-art BRB algorithms:
\begin{itemize}
    \item \textbf{kPlexS}~\footnote{\url{https://github.com/LijunChang/Maximum-kPlex}}: the existing algorithm proposed in~\cite{CXS22}.
    \item \textbf{KPLEX}~\footnote{\url{https://github.com/joey001/kplex\_degen\_gap}}: the existing algorithm proposed in~\cite{wang2023fast}.
    \item \textbf{DiseMKP}~\footnote{\url{https://github.com/huajiang-ynu/ijcai23-kpx}}: the existing algorithm proposed in~\cite{jiang2023refined}.
    \item \textbf{kPlexT}~\footnote{\url{https://github.com/LijunChang/Maximum-kPlex-v2}}: the existing algorithm proposed in~\cite{chang2024maximum}.
\end{itemize}

\smallskip
\noindent
\textbf{Setup.}
All algorithms are written in C++ and compiled with -O3 optimization by g++ 9.4.0.  
Moreover, all algorithms are initialized with a lower bound of $2k-2$ to focus on finding $k$-plexes with at least $2k-1$ vertices.
All experiments are conducted in the single-thread mode on a machine with an Intel(R) Xeon(R) Platinum 8358P CPU@2.60GHz and 256GB main memory. 
The CPU frequency is fixed at 3.3GHz.
We set the time limit as 3600 seconds and use \textbf{OOT} (Out Of Time limit) to indicate the time exceeds the limit. 
We consider six different values of $k$, i.e., $2,3,5,10,15$, and $20$. 
We focus on the case of $k=5$, and defer the experiments for other values of $k$ to the appendix. We also note that the major findings for $k=5$ hold for other values of $k$.

\smallskip
\noindent\textbf{Datasets.}
We consider the following two collections of graphs.
\begin{itemize}
    \item \textbf{Network Repository~\cite{NetworkRepo}.}
    The dataset contains 584 graphs with up to $5.87 \times 10 ^7$ vertices,
    including biological networks (36), dynamic networks (85), labeled networks (104), road networks (15), interaction (29),  scientific computing (11), social networks (75), facebook (114), web (31), and DIMACS-10 graphs (84).
    Most of them are real-world graphs.
    
    \item \textbf{2nd-DIMACS (DIMACS-2) Graphs~\cite{DIMACS}.}
    The dataset contains 80 synthetic dense graphs with up to 4000 vertices
    and the densities ranging from 0.03 to 0.99. Most graphs in the dataset are synthetic graphs, which are often hard to be solved~\cite{JZX+21,wang2023fast,jiang2023refined}.
\end{itemize}

For better comparisons, we select 30 representative graphs from the above 664 graphs and report the statistics in Table~\ref{table:graph-statistics}, where the graph density is $\frac{2m}{n(n-1)}$ and the maximum degree is $d_{max}$. 
The criteria of selecting these representative graphs are as follows. First, following~\cite{wang2023fast}, we do not select extremely easy or hard graphs, i.e., those graphs that can be solved within 5 seconds by all five solvers or cannot be solved within 3600 seconds by any solver when $k=5$. Second, the representative graphs cover a wide range of sizes. Among the selected graphs, 10 small dense graphs (G1-G10) are synthetic graphs from \textbf{DIMACS-2 Graphs}, 10 medium graphs (G11-G20) with at most $10^6$ vertices, and 10 large sparse graphs (G21-G30) with at least $10^6$ vertices are real-world graphs from \textbf{Network Repository}. Third, most of the representative graphs have also been selected in previous studies~\cite{JZX+21,wang2023fast,XLD+17}.

\begin{table}[t]
    \centering
    \small
    \caption{Statistics of 30 representative graphs}
    \label{table:graph-statistics}
    \scalebox{0.82}{
    \begin{tabular}{l|l|c|c|c|c|c}
    \hline
        \textbf{ID} & \textbf{Graph} & \textbf{$n$} & \textbf{$m$} & \textbf{density} & \textbf{$d_{max}$} & \textbf{$\delta(G)$} \\ \hline
        G1 & johnson8-4-4 & 70 & 1855 & $7.68 \cdot 10^{-1}$ & 53 & 53 \\ \hline
        G2 & C125-9 & 125 & 6963 & $8.98 \cdot 10^{-1}$ & 119 & 102 \\ \hline
        G3 & keller4 & 171 & 9435 & $6.49 \cdot 10^{-1}$ & 124 & 102 \\ \hline
        G4 & brock200-2 & 200 & 9876 & $4.96 \cdot 10^{-1}$ & 114 & 84 \\ \hline
        G5 & san200-0-9-1 & 200 & 17910 & $9.00 \cdot 10^{-1}$ & 191 & 162 \\ \hline
        G6 & san200-0-9-2 & 200 & 17910 & $9.00 \cdot 10^{-1}$ & 188 & 169 \\ \hline
        G7 & san200-0-9-3 & 200 & 17910 & $9.00 \cdot 10^{-1}$ & 187 & 169 \\ \hline
        G8 & p-hat300-1 & 300 & 10933 & $2.44 \cdot 10^{-1}$ & 132 & 49 \\ \hline
        G9 & p-hat300-2 & 300 & 21928 & $4.89 \cdot 10^{-1}$ & 229 & 98 \\ \hline
        G10 & p-hat500-1 & 500 & 31569 & $2.53 \cdot 10^{-1}$ & 204 & 86 \\ \hline
        G11 & soc-BlogCatalog-ASU & 10312 & 333983 & $6.28 \cdot 10^{-3}$ & 3992 & 114 \\ \hline
        G12 & socfb-UIllinois & 30795 & 1264421 & $2.67 \cdot 10^{-3}$ & 4632 & 85 \\ \hline
        G13 & soc-themarker & 69413 & 1644843 & $6.83 \cdot 10^{-4}$ & 8930 & 164 \\ \hline
        G14 & soc-BlogCatalog & 88784 & 2093195 & $5.31 \cdot 10^{-4}$ & 9444 & 221 \\ \hline
        G15 & soc-buzznet & 101163 & 2763066 & $5.40 \cdot 10^{-4}$ & 64289 & 153 \\ \hline
        G16 & soc-LiveMocha & 104103 & 2193083 & $4.05 \cdot 10^{-4}$ & 2980 & 92 \\ \hline
        G17 & soc-wiki-conflict & 116836 & 2027871 & $2.97 \cdot 10^{-4}$ & 20153 & 145 \\ \hline
        G18 & soc-google-plus & 211187 & 1141650 & $5.12 \cdot 10^{-5}$ & 1790 & 135 \\ \hline
        G19 & soc-FourSquare & 639014 & 3214986 & $1.57 \cdot 10^{-5}$ & 106218 & 63 \\ \hline
        G20 & rec-epinions-user-ratings & 755760 & 13667951 & $4.79 \cdot 10^{-5}$ & 162179 & 151 \\ \hline
        G21 & soc-wiki-Talk-dir & 1298165 & 2288646 & $2.72 \cdot 10^{-6}$ & 100025 & 119 \\ \hline
        G22 & soc-pokec & 1632803 & 22301964 & $1.67 \cdot 10^{-5}$ & 14854 & 47 \\ \hline
        G23 & tech-ip & 2250498 & 21643497 & $8.55 \cdot 10^{-6}$ & 1833161 & 253 \\ \hline
        G24 & ia-wiki-Talk-dir & 2394385 & 4659565 & $1.63 \cdot 10^{-6}$ & 100029 & 131 \\ \hline
        G25 & sx-stackoverflow & 2584164 & 28183518 & $8.44 \cdot 10^{-6}$ & 44065 & 198 \\ \hline
        G26 & web-wikipedia\_link\_it & 2790239 & 86754664 & $2.23 \cdot 10^{-5}$ & 825147 & 894 \\ \hline
        G27 & socfb-A-anon & 3097165 & 23667394 & $4.93 \cdot 10^{-6}$ & 4915 & 74 \\ \hline
        G28 & soc-livejournal-user-groups & 7489073 & 112305407 & $4.00 \cdot 10^{-6}$ & 1053720 & 116 \\ \hline
        G29 & soc-bitcoin & 24575382 & 86063840 & $2.85 \cdot 10^{-7}$ & 1083703 & 325 \\ \hline
        G30 & soc-sinaweibo & 58655849 & 261321033 & $1.52 \cdot 10^{-7}$ & 278489 & 193 \\ \hline
    \end{tabular}
    }
\end{table}

\subsection{Comparing with State-of-the-art Algorithms}
\noindent
\textbf{Number of solved instances on two collections of graphs.} 
We compare \texttt{kPEX} with four baselines by reporting the numbers of solved instances. 
The results for {\bf Network Repository} are shown in Table~\ref{table:num-instance-Network-Repository} and Figure~\ref{img-number-instances-Network-Repository}. We observe that \texttt{kPEX} outperforms all baselines for all tested values of $k$. For example, \texttt{kPEX} solves 12 instances more than the best baseline \texttt{KPLEX} for $k=10$ within 3600 seconds. 
In addition, our \texttt{kPEX} is more stable than baselines when varying $k$. In contrast, there is an obvious drop in solved instances for \texttt{kPlexT} and \texttt{DiseMKP} as $k$ increases from 2 to 20.
This demonstrates the superiority of \texttt{kPEX}, which employs the \texttt{AltRB} strategy (with novel reduction and bounding techniques) and efficient pre-processing methods.
The results on the collection of {\bf DIMACS-2 Graphs} are shown in Table~\ref{table:num-instance-artificial} and Figure~\ref{img-number-instances-artificial}. \texttt{kPEX} outperforms all baselines by solving the most instances with 3600 seconds for all tested $k$ values, 
e.g., \texttt{kPEX} solves 9 instances more than the second best solver \texttt{KPLEX} for $k=5$.
Besides, we note that \texttt{kPEX} is comparable with \texttt{DiseMKP} when $k=2$. This is because the proposed reduction rules and upper bounding method are less effective for small values of $k$.

\begin{table}[t]
    \centering
    \caption{Number of solved instances on Network Repository within 3600 seconds}
    \label{table:num-instance-Network-Repository}
    \begin{tabular}{cccccc}
    \toprule
        \textbf{$k$} & \textbf{\texttt{kPEX} (ours)} & \textbf{\texttt{KPLEX}} & \textbf{\texttt{kPlexT}} & \textbf{\texttt{kPlexS}} & \textbf{\texttt{DiseMKP}} \\ 
    \midrule
        2 & \textbf{567} & 564 & 559 & 559 & 542 \\ 
        3 & \textbf{564} & 553 & 557 & 553 & 527 \\ 
        5 & \textbf{565} & 557 & 554 & 547 & 516 \\ 
        10 & \textbf{564} & 552 & 537 & 549 & 495 \\ 
        15 & \textbf{559} & 548 & 507 & 547 & 452 \\ 
        20 & \textbf{559} & 546 & 471 & 539 & 439 \\
    \bottomrule
    \end{tabular}
\end{table}
\begin{figure}[t]
    \centering
    \subfigure[$k$=2]{
        \includegraphics[width=0.225\textwidth]{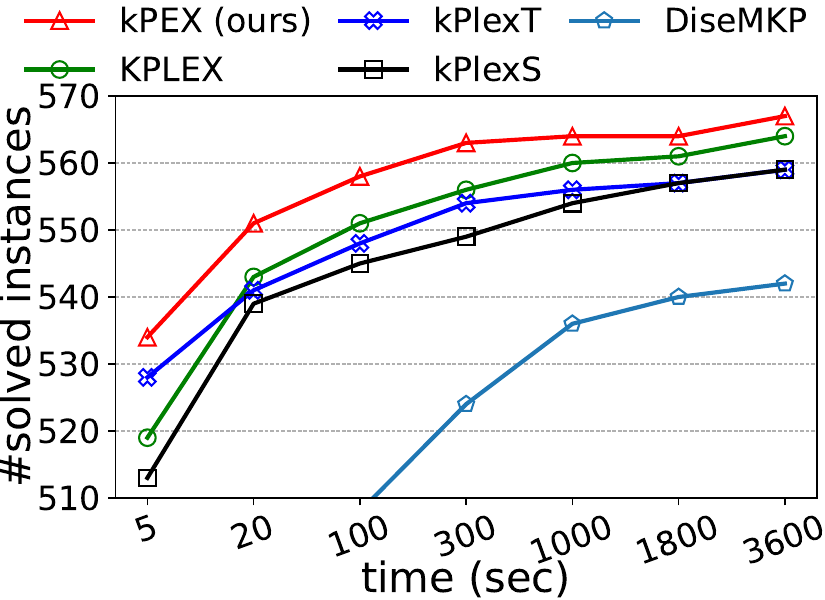}
    }
    \subfigure[$k$=3]{
        \includegraphics[width=0.225\textwidth]{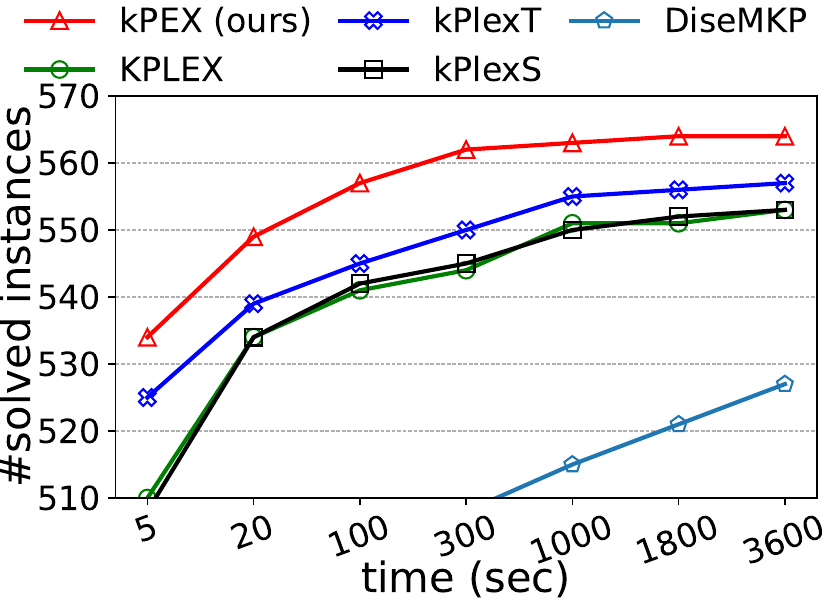}
    }
    \subfigure[$k$=5]{
        \includegraphics[width=0.225\textwidth]{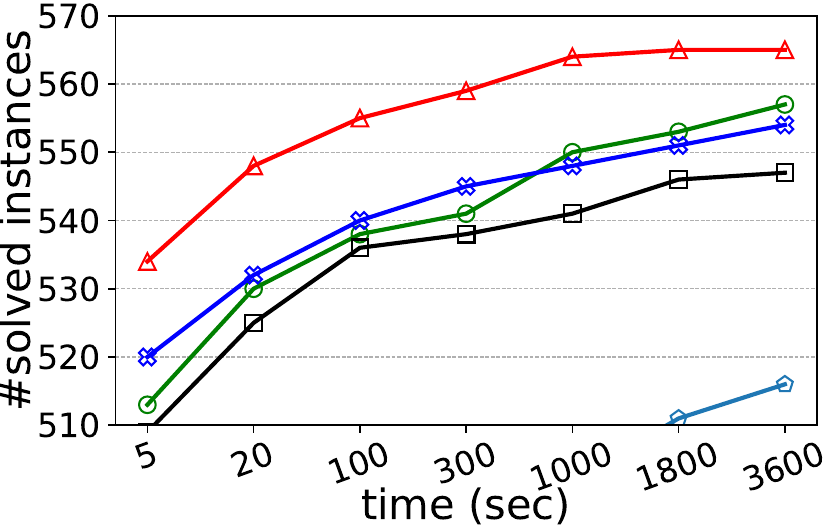}
    }
    \subfigure[$k$=10]{
        \includegraphics[width=0.225\textwidth]{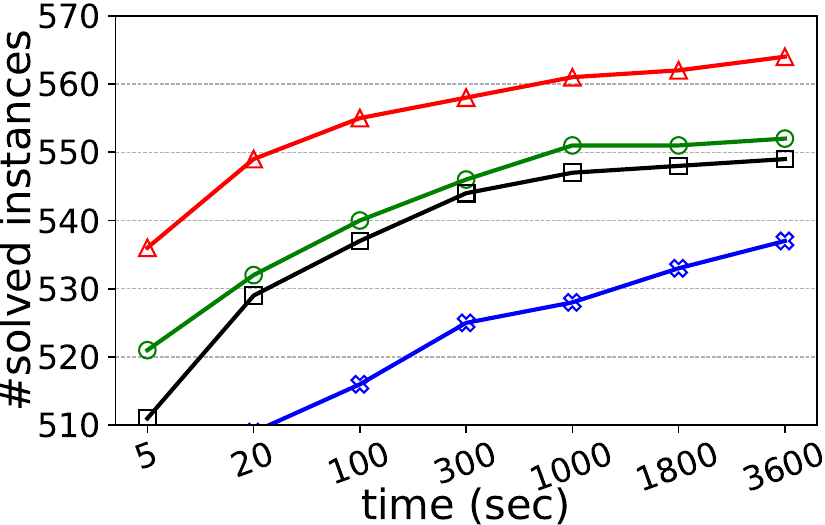}
    }
    \subfigure[$k$=15]{
        \includegraphics[width=0.225\textwidth]{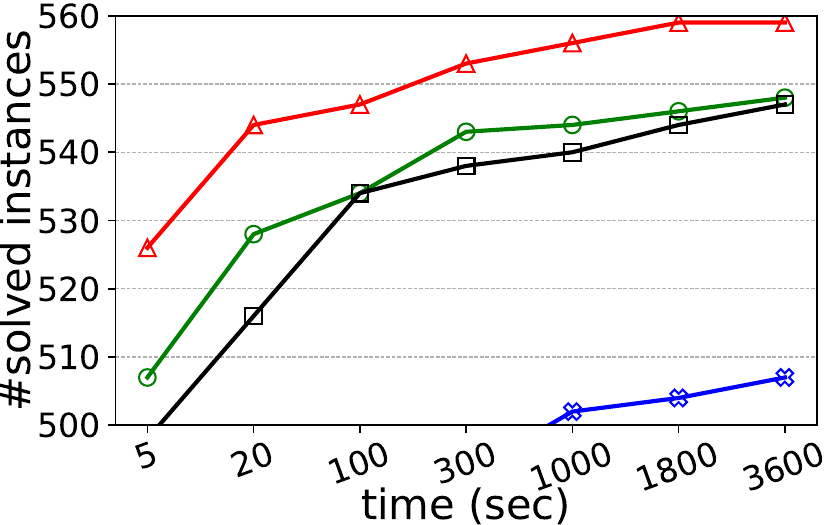}
    }
    \subfigure[$k$=20]{
        \includegraphics[width=0.225\textwidth]{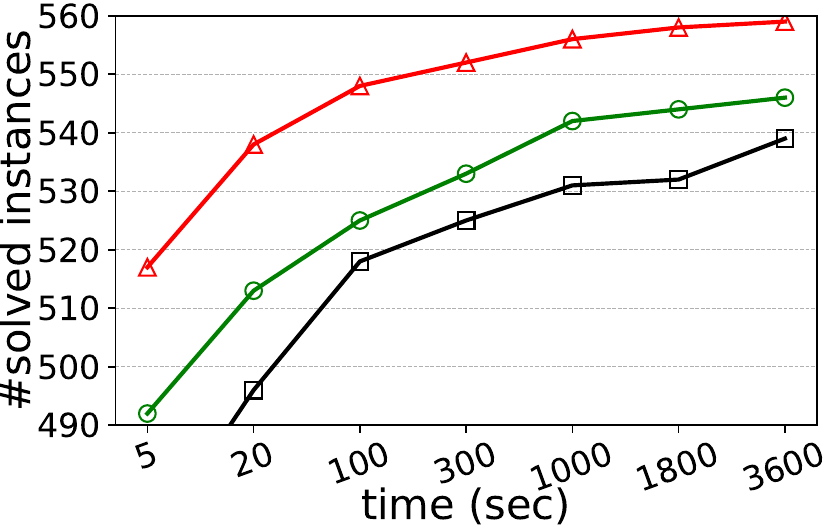}
    }
    \caption{Number of solved instances on Network Repository (The lines corresponding to \texttt{DiseMKP} and \texttt{kPlexT} may not appear in the figures, as they are slow under certain settings and thus cannot reach the bottom lines within 3600 seconds.)}
    \label{img-number-instances-Network-Repository}
\end{figure}

\begin{table}[t]
    \centering
    \caption{Number of solved instances on DIMACS-2 within 3600 seconds}
    \label{table:num-instance-artificial}
    \begin{tabular}{cccccc}
    \toprule
        \textbf{$k$} & \textbf{\texttt{kPEX} (ours)} & \textbf{\texttt{KPLEX}} & \textbf{\texttt{kPlexT}} & \textbf{\texttt{kPlexS}} & \textbf{\texttt{DiseMKP}} \\ 
    \midrule
        2 & \textbf{29} & 27 & 25 & 22 & 27 \\ 
        3 & \textbf{28} & 23 & 24 & 20 & 25 \\ 
        5 & \textbf{27} & 18 & 17 & 15 & 17 \\ 
        10 & \textbf{22} & 17 & 14 & 15 & 16 \\ 
        15 & \textbf{23} & 22 & 20 & 20 & 21 \\ 
        20 & \textbf{26} & 21 & 21 & 21 & 18 \\
    \bottomrule
    \end{tabular}
\end{table}
\begin{figure}[t]
    \centering
    \subfigure[$k$=2]{
        \includegraphics[width=0.225\textwidth]{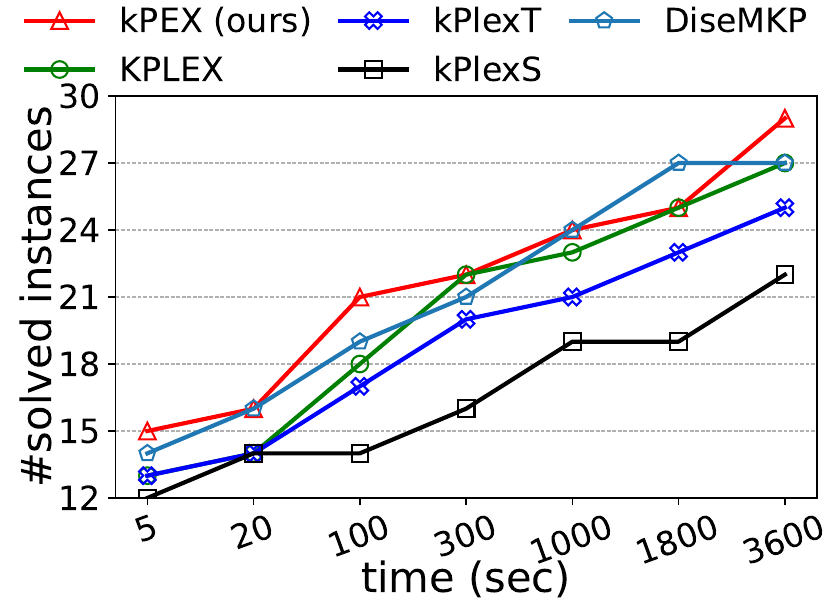}
    }
    \subfigure[$k$=3]{
        \includegraphics[width=0.225\textwidth]{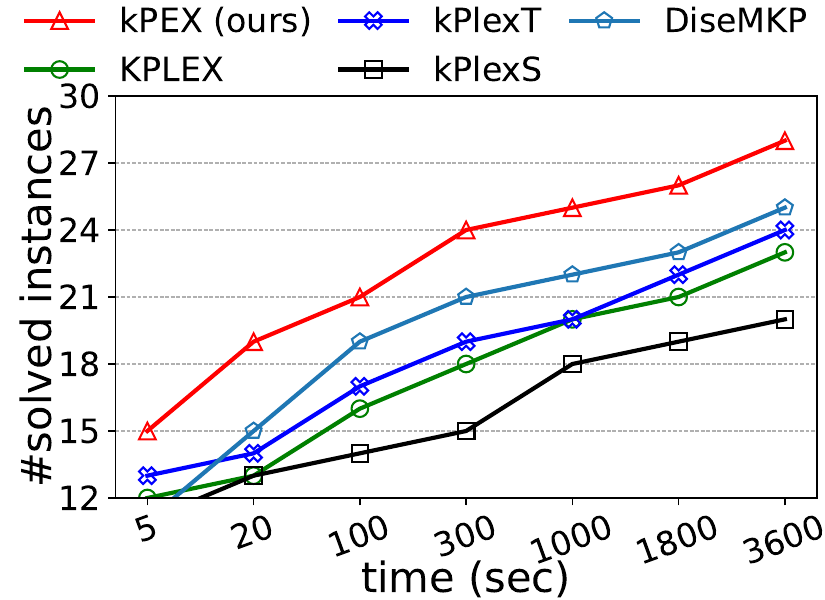}
    }
    \subfigure[$k$=5]{
        \includegraphics[width=0.225\textwidth]{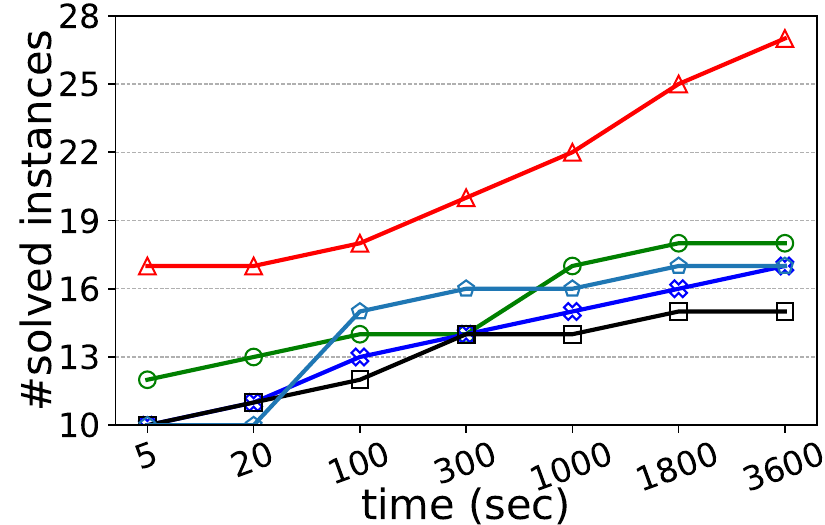}
    }
    \subfigure[$k$=10]{
        \includegraphics[width=0.225\textwidth]{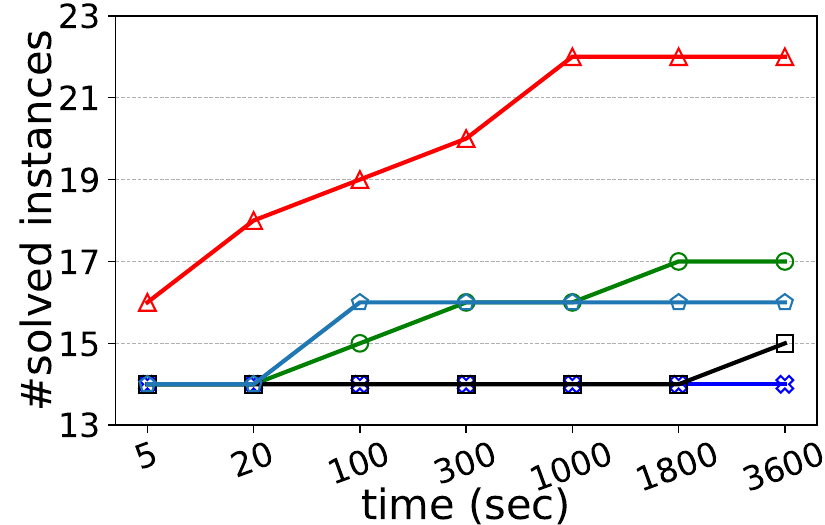}
    }
    \subfigure[$k$=15]{
        \includegraphics[width=0.225\textwidth]{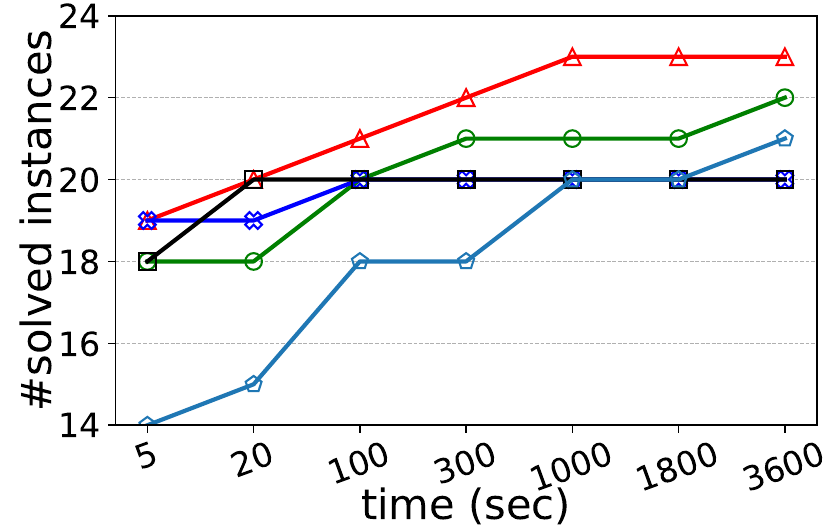}
    }
    \subfigure[$k$=20]{
        \includegraphics[width=0.225\textwidth]{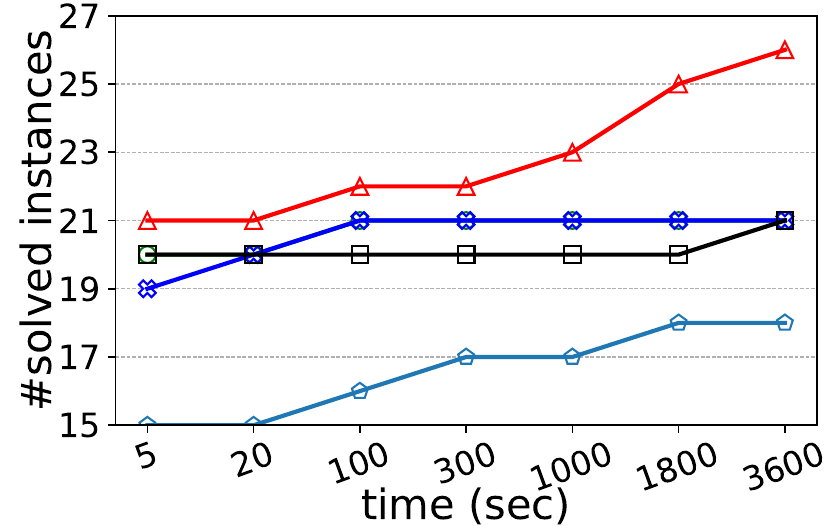}
    }
    \caption{Number of solved instances on DIMACS-2}
    \label{img-number-instances-artificial}
\end{figure}

\smallskip
\noindent
\textbf{Running times on representative graphs.}
We report the running times of all algorithms on 30 representative graphs with $k=5$ in Table~\ref{table:comparing-algorithms-k=5}.
We observe that \texttt{kPEX} outperforms all baselines by achieving significant speedups on the majority graphs. For example, \texttt{kPEX} runs at least 5 times faster than \texttt{KPLEX} on 25 out of 30 graphs and at least 5 times faster than \texttt{kPlexT} on 21 out of 30 graphs. Moreover, there are 7 out of 30 graphs where \texttt{kPEX} runs at least 100 times faster than all baselines. 
Note that \texttt{kPEX} may exhibit slower performance compared to baselines on rare occasions. For instance, the baseline \texttt{kPlexT} runs faster than our \texttt{kPEX} on G23 with $k=5$. The possible reasons include: 1) All algorithms rely on some heuristic procedures, e.g., the heuristic method for finding a large initial $k$-plex. The performance of these heuristic methods varies across different graphs and settings; 2) Compared with baselines, \texttt{kPEX} incorporates newly proposed reduction techniques, which may introduce additional time costs.

\begin{table}[t]
    \centering
    \small
    \caption{Running time in seconds of \texttt{kPEX} and state-of-the-arts on 30 graphs with $k=5$}
    \begin{tabular}{r|r|rrrr}
    \hline
        \textbf{ID} & \textbf{\texttt{kPEX} (ours)} & \textbf{\texttt{KPLEX}} & \textbf{\texttt{kPlexT}} & \textbf{\texttt{kPlexS}} & \textbf{\texttt{DiseMKP}} \\ \hline
        G1 & \textbf{3.88} & 1154.76 & 120.63 & 187.88 & 55.56 \\ 
        G2 & \textbf{1360.94} & OOT & OOT & OOT & OOT \\ 
        G3 & \textbf{1527.41} & OOT & OOT & OOT & OOT \\ 
        G4 & \textbf{306.14} & OOT & 3160.85 & OOT & OOT \\ 
        G5 & \textbf{0.10} & 0.29 & 34.09 & 1356.45 & 0.16 \\ 
        G6 & \textbf{24.63} & OOT & OOT & OOT & OOT \\ 
        G7 & \textbf{433.86} & OOT & OOT & OOT & OOT \\ 
        G8 & \textbf{2.22} & 567.84 & 20.91 & OOT & 27.92 \\ 
        G9 & \textbf{2.82} & 411.26 & OOT & OOT & OOT \\ 
        G10 & \textbf{191.31} & OOT & 1339.31 & OOT & 1133.52 \\ \hline
        G11 & \textbf{3.97} & 1766.68 & 2318.35 & OOT & OOT \\ 
        G12 & \textbf{0.39} & 1.30 & 0.53 & 1.36 & 440.90 \\ 
        G13 & \textbf{55.29} & OOT & OOT & OOT & OOT \\ 
        G14 & \textbf{927.11} & OOT & OOT & OOT & OOT \\ 
        G15 & \textbf{21.17} & OOT & OOT & OOT & OOT \\ 
        G16 & \textbf{1.68} & 58.82 & 27.28 & 1655.40 & 900.80 \\ 
        G17 & \textbf{1.68} & 3022.10 & 123.86 & OOT & OOT \\ 
        G18 & \textbf{0.72} & 2804.46 & 1818.90 & 1099.39 & OOT \\ 
        G19 & \textbf{0.64} & 3.59 & 2.15 & 2.06 & 1695.05 \\ 
        G20 & \textbf{4.39} & 795.20 & 204.00 & 72.03 & 1347.27 \\ 
        G21 & \textbf{2.82} & 961.06 & 1515.40 & OOT & OOT \\ 
        G22 & \textbf{2.42} & 13.58 & 3.72 & 11.77 & 18.25 \\ 
        G23 & 13.16 & 136.61 & \textbf{4.80} & 11.39 & OOT \\ 
        G24 & \textbf{5.61} & 2979.37 & 3055.73 & OOT & OOT \\ 
        G25 & \textbf{3.80} & 92.05 & 203.26 & OOT & OOT \\ 
        G26 & \textbf{4.84} & 700.66 & 7.25 & 39.25 & 8.41 \\ 
        G27 & \textbf{2.60} & 14.52 & 3.50 & 15.39 & 51.94 \\ 
        G28 & \textbf{139.52} & OOT & OOT & 1703.49 & OOT \\ 
        G29 & \textbf{6.08} & 312.93 & OOT & OOT & OOT \\ 
        G30 & \textbf{593.26} & OOT & OOT & OOT & OOT \\   \hline
    \end{tabular}
    \label{table:comparing-algorithms-k=5}
\end{table}

\subsection{Effectiveness of Proposed Techniques}
We compare the running time of \texttt{kPEX} with its variants:

\begin{itemize}
    \item  \textbf{\texttt{kPEX-SeqRB}}: \texttt{kPEX} replaces \texttt{AltRB} with \texttt{SeqRB}(Section~\ref{sec:altrb}).
    \item \textbf{\texttt{kPEX-CTCP}}: \texttt{kPEX} replaces \texttt{CF-CTCP} with \texttt{CTCP} (\cite{CXS22}).
    \item \textbf{\texttt{kPEX-EGo}}: \texttt{kPEX} replaces \texttt{KPHeuris} with the existing heuristic method \texttt{EGo-Degen} in~\cite{chang2024maximum}.
    \item \textbf{\texttt{kPEX-Degen}}: \texttt{kPEX} replaces \texttt{KPHeuris} with the existing heuristic method \texttt{kPlex-Degen} in~\cite{CXS22}.
\end{itemize}
In other words, \texttt{kPEX-SeqRB} is the version without \texttt{AltRB}; \texttt{kPEX-CTCP} is the version without \texttt{CF-CTCP}; \texttt{kPEX-EGo} and \texttt{kPEX-Degen} are the versions without \texttt{KPHeuris}.

\begin{table}[!ht]
    \centering
    \small
    \caption{Running time in seconds of \texttt{kPEX} and its variants on 30 graphs with $k=5$}
    \label{table:ablation-k=5}
    \begin{tabular}{c|c|cccc}
    \hline
        \textbf{ID} & \textbf{\texttt{kPEX}} & \textbf{\texttt{kPEX-SeqRB}} & \textbf{\texttt{kPEX-CTCP}} & \textbf{\texttt{kPEX-EGo}} & \textbf{\texttt{kPEX-Degen}} \\ 
    \hline
        G1 & \textbf{3.88} & 18.66 & 3.89 & 3.91 & 3.92 \\ 
        G2 & \textbf{1360.94} & OOT & 1364.70 & 1371.95 & 1365.94 \\ 
        G3 & 1527.41 & OOT & \textbf{1525.00} & 1539.98 & 1538.84 \\ 
        G4 & 306.14 & 3220.90 & \textbf{305.95} & 307.76 & 307.69 \\ 
        G5 & 0.10 & 0.11 & 0.10 & 0.10 & \textbf{0.09} \\ 
        G6 & \textbf{24.63} & 557.07 & 24.63 & 33.37 & 57.68 \\ 
        G7 & \textbf{433.86} & OOT & 434.65 & 528.89 & 674.42 \\ 
        G8 & 2.22 & 18.54 & 2.22 & \textbf{2.19} & 2.20 \\ 
        G9 & 2.82 & 20.47 & \textbf{2.81} & 4.21 & 4.18 \\ 
        G10 & \textbf{191.31} & 3472.65 & 191.72 & 202.50 & 202.49 \\ \hline
        G11 & 3.97 & 23.51 & 4.63 & 4.19 & \textbf{3.87} \\ 
        G12 & 0.39 & \textbf{0.39} & 2.24 & 0.88 & 0.71 \\ 
        G13 & \textbf{55.29} & 555.36 & 59.38 & 97.39 & 103.23 \\ 
        G14 & \textbf{927.11} & OOT & 936.52 & 1113.52 & 1197.87 \\ 
        G15 & \textbf{21.17} & 128.27 & 35.76 & 29.85 & 24.84 \\ 
        G16 & \textbf{1.68} & 2.19 & 5.84 & 2.35 & 2.18 \\ 
        G17 & \textbf{1.68} & 8.49 & 8.27 & 3.88 & 1.90 \\ 
        G18 & 0.72 & 6.76 & 1.34 & 0.90 & \textbf{0.66} \\ 
        G19 & \textbf{0.64} & 0.66 & 17.41 & 8.36 & 93.70 \\ 
        G20 & \textbf{4.39} & 4.46 & 222.02 & 150.66 & 260.41 \\ 
        G21 & \textbf{2.82} & 8.03 & 3.55 & 3.68 & 3.32 \\ 
        G22 & 2.42 & \textbf{2.26} & 16.82 & 8.28 & 2.40 \\ 
        G23 & 13.16 & \textbf{12.62} & 1618.38 & 461.19 & 117.33 \\ 
        G24 & \textbf{5.61} & 15.80 & 6.97 & 9.83 & 9.03 \\ 
        G25 & 3.80 & \textbf{3.56} & 66.56 & 22.87 & 3.69 \\ 
        G26 & 4.84 & 4.77 & 4.64 & 5.61 & \textbf{4.54} \\ 
        G27 & 2.60 & \textbf{2.48} & 24.83 & 11.07 & 2.63 \\ 
        G28 & \textbf{139.52} & 139.96 & OOT & OOT & OOT \\ 
        G29 & 6.08 & 17.57 & \textbf{5.99} & 11.18 & 10.85 \\ 
        G30 & \textbf{593.26} & OOT & 1628.00 & 2094.29 & 2020.66 \\ 
    \hline
    \end{tabular}
\end{table}

\smallskip
\noindent
\textbf{Effectiveness of \texttt{AltRB}.} 
We compare \texttt{kPEX} with \texttt{kPEX-SeqRB} and report the running times in Table~\ref{table:ablation-k=5}. 
We observe that \texttt{kPEX} performs better than \texttt{kPEX-SeqRB} by achieving at least a 5$\times$ speedup on 12 out of 30 graphs and running up to 20$\times$ faster on G6. 
This indicates the effectiveness of \texttt{AltRB} in narrowing down the search space. Besides, \texttt{AltRB} contributes more speedups on synthetic graphs G1-G10 since the running time is dominated by the branch-reduction-and-bound stage on these graphs.

\begin{table}[!ht]
    \caption{Pre-processing time in seconds on 20 graphs with $k$=5
    ($lb$ denotes the size of the computed heuristic $k$-plex)
    }
    \label{table:compare-prepro-k=5}
    \centering
    \small
    \begin{tabular}{r|cc|cc|cc|cc}
    \hline
    \multirow{2}{*}{\textbf{ID}} & \multicolumn{2}{c|}{\textbf{\texttt{kPEX}}} & \multicolumn{2}{c|}{\textbf{\texttt{kPlexT}}} & \multicolumn{2}{c|}{\textbf{\texttt{kPlexS}}} & \multicolumn{2}{c}{\textbf{\texttt{DiseMKP}}} \\
     & \textbf{time} & \textbf{$lb$} & \textbf{time} & \textbf{$lb$} & \textbf{time} & \textbf{$lb$} & \textbf{time} & \textbf{$lb$} \\ \hline
        G11 & 0.43 & \textbf{51} & 0.41 & 50 & \textbf{0.30} & 50 & 0.53 & 49 \\ 
        G12 & \textbf{0.39} & \textbf{73} & 0.58 & 69 & 1.30 & 34 & 1.71 & 34 \\ 
        G13 & 4.35 & \textbf{39} & \textbf{1.69} & 37 & 2.44 & 36 & 3.85 & 35 \\ 
        G14 & 9.24 & \textbf{70} & 5.01 & 67 & \textbf{3.08} & 64 & 6.05 & 66 \\ 
        G15 & 3.06 & \textbf{49} & \textbf{2.06} & 48 & 4.08 & 48 & 10.59 & 47 \\ 
        G16 & 1.29 & \textbf{27} & \textbf{0.82} & 26 & 1.84 & 14 & 3.85 & 14 \\ 
        G17 & 0.68 & \textbf{39} & \textbf{0.62} & 37 & 1.69 & 37 & 5.53 & 37 \\ 
        G18 & \textbf{0.14} & \textbf{87} & 0.34 & \textbf{87} & 0.27 & \textbf{87} & 0.51 & \textbf{87} \\ 
        G19 & \textbf{0.63} & \textbf{44} & 2.08 & 42 & 0.93 & 37 & 9.60 & 37 \\ 
        G20 & \textbf{3.48} & \textbf{21} & 33.60 & 19 & 32.56 & 10 & 118.24 & 10 \\ 
        G21 & 0.64 & \textbf{44} & \textbf{0.37} & 43 & 0.43 & 43 & 0.66 & 42 \\ 
        G22 & \textbf{2.42} & \textbf{34} & 4.51 & 32 & 15.46 & 26 & 15.55 & 27 \\ 
        G23 & 13.16 & \textbf{11} & \textbf{5.91} & \textbf{11} & 13.05 & 10 & 1708.75 & 10 \\ 
        G24 & 1.28 & \textbf{44} & \textbf{0.67} & 41 & 0.94 & 41 & 1.40 & 41 \\ 
        G25 & \textbf{3.14} & \textbf{77} & 7.90 & 76 & 34.82 & 76 & 45.58 & \textbf{77} \\ 
        G26 & \textbf{2.90} & \textbf{881} & 3.57 & \textbf{881} & 3.94 & \textbf{881} & 3.90 & 880 \\ 
        G27 & \textbf{2.60} & \textbf{37} & 4.79 & 35 & 18.02 & 32 & 21.06 & 33 \\ 
        G28 & \textbf{102.57} & \textbf{17} & 719.41 & 15 & 1210.59 & 12 & OOT & - \\ 
        G29 & \textbf{2.74} & \textbf{296} & 3.17 & 292 & 3.31 & 292 & 8.52 & 292 \\ 
        G30 & \textbf{102.73} & \textbf{65} & 134.74 & 62 & 252.29 & 17 & 1217.99 & 16 \\   \hline
    \end{tabular}
\end{table}

\smallskip
\noindent
\textbf{Effectiveness of \texttt{CF-CTCP}.}
We compare \texttt{kPEX} with \texttt{kPEX-CTCP}, and the running times are reported in Table~\ref{table:ablation-k=5}. First, \texttt{kPEX} and \texttt{kPEX-CTCP} have similar performance on G1-G10 because the pre-processing techniques take little time (e.g., less than 1 second) on these synthetic graphs. Second, \texttt{kPEX} runs at least 5 times faster than \texttt{kPEX-CTCP} on 8 out of 20 real-world graphs. Moreover, \texttt{CF-CTCP} provides at least 50$\times$ speedup on G20 and G23. These results show the effectiveness of \texttt{CF-CTCP} on large sparse graphs.

\smallskip
\noindent
\textbf{Effectiveness of \texttt{KPHeuris}.} 
We compare \texttt{kPEX} with its variants \texttt{kPEX-EGo} and \texttt{kPEX-Degen} (note that \texttt{CF-CTCP} is not replaced). The running times are shown in Table~\ref{table:ablation-k=5}. 
We have the following observations. First, the running time of \texttt{kPEX} is less than that of both variants on the majority of graphs (i.e., on 23 out of 30 graphs). Then, \texttt{kPEX} runs at least 
5 times faster than \texttt{kPEX-EGo} on 5 out of 30 graphs and faster than \texttt{kPEX-Degen} on 4 out of 30 graphs. In addition, \texttt{kPEX} runs at least 25 times faster than both \texttt{kPEX-EGo} and \texttt{kPEX-Degen} on G20 and G28.
This shows that making more effort to finding a larger initial $k$-plex benefits \texttt{kPEX} by narrowing down the search space.
Second, although \texttt{kPEX} may be slightly slower than the two variants, the extra time consumption is small and can be ignored compared to the total running time. For example, \texttt{kPEX} is 0.1 seconds slower than \texttt{kPEX-Degen} on G11 due to the extra computation, while the total running time of \texttt{kPEX} is 3.97 seconds, which means that the extra time consumption is negligible.
Third, the performance of \texttt{kPEX-EGo} and \texttt{kPEX-Degen} is better than \texttt{kPEX-SeqRB} on G1-G10. This means that the variant of \texttt{kPEX} without \texttt{AltRB} is slower than the variant without \texttt{KPHeuris}.
This indicates that \texttt{AltRB} provides a greater performance boost than heuristic techniques on those graphs where branch-reduction-and-bound stage dominates the running time.

\smallskip
\noindent
\textbf{Effectiveness of \texttt{KPHeuris} and \texttt{CF-CTCP}.}
We also compare the total pre-processing time and the size of the $k$-plex (i.e., $lb$) obtained by different heuristic methods in \texttt{kPEX}, \texttt{kPlexT}, \texttt{kPlexS}, and \texttt{DiseMKP} (note that \texttt{KPLEX} uses the same pre-processing method as \texttt{kPlexS}).
The results are reported in Table~\ref{table:compare-prepro-k=5}.
Note that we exclude the results on synthetic graphs G1-G10 since they have only hundreds of vertices and can be handled within 1 second by all methods.
We have the following observations.
First, \texttt{kPEX} consistently obtains the largest $lb$ (or matches the largest obtained by others) while the pre-processing time remains comparable to other algorithms.
Second, \texttt{KPHeuris} outperforms the other pre-processing algorithms by obtaining a lager $k$-plex while costing much less time on G20 and G28. This also verifies the effectiveness of \texttt{CF-CTCP} and \texttt{KPHeuris}.

\section{Related work}\label{sec:related-work}
\noindent
\textbf{Maximum $k$-plex search.}
The \emph{maximum $k$-plex search} problem has garnered significant attention in social network analysis~\cite{MH12,MNS12} since the concept of $k$-plex was first proposed in~\cite{seidman1978graph}. 
Balasundaram et al.~\cite{BBH11} showed the NP-hardness of the problem with any fixed $k$.
Consequently, the major algorithmic design paradigm for exact solution is based on the {\em branch-reduction-and-bound} (BRB) framework~\cite{XLD+17,wang2023fast,GCY+18,zhou2021improving,JZX+21,CXS22,jiang2023refined,chang2024maximum}.
In particular, Xiao et al.~\cite{XLD+17} proposed a branching strategy, which improves theoretical time complexity from the trivial bound of $O^*(2^n)$ to $O^*(c^n)$ where $c<2$ and $O^*$ ignores polynomial factors. Later, Wang et al.~\cite{wang2023fast} designed \textbf{KPLEX} which is parameterized by the degeneracy gap (bounded empirically by $O(\log n)$).
Very recently, Chang and Yao~\cite{chang2024maximum} proposed \textbf{kPlexT}, which improves the worst-case time complexity with newly proposed branching and reduction techniques.
Additionally, several reduction and bounding techniques have been designed in the BRB framework to boost the practical performance.
Gao et al.~\cite{GCY+18} developed reduction methods and a dynamic vertex selection strategy.
Later, Zhou et al.~\cite{zhou2021improving} proposed a stronger reduction method and designed a coloring-based bounding method.
Jiang et al.~\cite{JZX+21} designed a partition-based bounding method, and later in~\cite{jiang2023refined}, their algorithm \textbf{DiseMKP} is equipped with a better upper bound. 
Chang et al.~\cite{CXS22} designed an efficient algorithm \textbf{kPlexS} with a novel reduction method \texttt{CTCP} and a heuristic method.
We note that the algorithms designed by Xiao et al.~\cite{XLD+17} and Chang and Yao~\cite{chang2024maximum} also work for the case when there is no requirement for the found $k$-plex to be of size at least $2k-1$.
We remark that existing works mainly focus on the BRB framework that conducts the reduction and the bounding sequentially, and our solution \texttt{kPEX} firstly adopts a new BRB framework that alternatively and iteratively conducts the reduction and the bounding.

\smallskip
\noindent
\textbf{Maximal $k$-plex enumeration.}
Another related problem is \emph{maximal $k$-plex enumeration}, which aims to list all all maximal $k$-plexes in the input graph; Here, a $k$-plex is \emph{maximal} if it cannot be contained in other $k$-plexes.
Many efficient algorithms are proposed for enumerating maximal $k$-plexes, including Bron-Kerbosch-based algorithms~\cite{WP07,WANG2017,CDD+18,CFM+17,DLQ+22,wang2022listing} and reverse-search-based algorithms~\cite{BCK15}. 
We remark that existing algorithms for enumerating maximal $k$-plexes  can be utilized to solve the studied problem by listing all maximal $k$-plexes and then returning the largest one among them (note that the maximum $k$-plex is the maximal $k$-plex with largest number of vertices).
However, the resulting solutions are not efficient due to the limited pruning and bounding techniques, as verified in ~\cite{chang2024maximum}.

\smallskip
\noindent
\textbf{Other cohesive subgraph models.}
$k$-plexes reduce to cliques when $k = 1$. There have been lines of work focusing on the maximum clique search and maximal clique enumeration problems~\cite{CP90,PardalosX94a,Tomita17,Cha19,conte2016finding,ELS13,NAUDE201628,TOMITA200628}.
Further, the concept of $k$-plex is also explored in other kinds of graphs, e.g., bipartite graphs~\cite{yu2022efficient,yu2022maximum,luo2022maximum,chen2021efficient,dai2023hereditary}, directed graphs~\cite{GYL+24}, temporal graphs~\cite{BHM+19}, uncertain graphs~\cite{dai2022uncertain}, and so on.
Besides $k$-plex, various cohesive subgraph models have been studied, including $k$-core~\cite{BVZ03m,CKC+11efficient}, $k$-truss~\cite{cohen2008trusses,huang2014querying,WC12}, $\gamma$-quasi-clique~\cite{pei2005mining,zeng2006coherent,khalil2022parallel,yu2023fast}, $k$-defective clique~\cite{chang2023efficient,dai2023maximal,gao2022exact,chen2021computing}, densest subgraph~\cite{ma2021directed,XMF+24}, and so on.
For an overview on cohesive subgraph search, we refer to excellent books and surveys~\cite{LRJ+10,CQ18,HLX19,FHQ+20survey,FW+21cohesive}.

\section{Conclusion}\label{sec:conclusion}
In this paper, we studied the maximum $k$-plex search problem. We proposed a new branch-reduction-and-bound method, called \texttt{kPEX}, which includes a new alternated reduction-and-bound process \texttt{AltRB}. In addition, we also designed efficient pre-processing techniques for boosting the performance, which includes \texttt{KPHeuris} for computing a large heuristic $k$-plex and \texttt{CF-CTCP} for efficiently removing unpromising vertices/edges. Extensive experiments on 664 graphs verified \texttt{kPEX}'s superiority over state-of-the-art algorithms. In the future, we will explore the possibility of adapting \texttt{kPEX} to mining other cohesive subgraphs.

\bibliographystyle{ACM-Reference-Format}
\bibliography{references}

\appendix

\section{Additional Descriptions of \texttt{CF-CTCP}}
\subsection{Time Complexity of \texttt{CF-CTCP}}

Before analyzing the time complexity of \texttt{CF-CTCP} (Algorithm~\ref{algorithm:CF-CTCP}), we first prove the following lemma.
\begin{lemma}\label{lemma:sum-min-deg-for-edges}
    Given a graph $G=(V, E)$, we have 
    \begin{equation*}
        \sum_{(u,v) \in E}\min(d_G(u), d_G(v)) \leq 2m \times \delta(G).
    \end{equation*}
\end{lemma}
\begin{proof}
    Assume that vertices $v_1, v_2, ..., v_n$ in $G$ are sorted according to the degeneracy order, indicating that $|N_G^+(v_i)|=|N_G(v_i) \cap \{v_{i+1}, v_{i+2}, ..., v_n\}| \leq \delta(G)$.
    Thus we have
    \begin{align*}
        & \sum_{(u,v) \in E}\min(d_G(u), d_G(v)) \\
        &= \sum_{v_i \in V} \sum_{v_j \in N_G^+(v_i)}\min(d_G(v_i), d_G(v_j)) \\
        &\leq \sum_{v_i \in V} \sum_{v_j \in N_G^+(v_i)} d_G(v_i) 
        \leq \sum_{v_i \in V} d_G(v_i) \times \delta(G) = 2m \times \delta(G).
    \end{align*}
\end{proof}

We can derive from Lemma~\ref{lemma:sum-min-deg-for-edges} that 
$$O(\sum_{(u,v) \in E}\min(d_G(u), d_G(v)))=O(m\times \delta(G)).$$ 
Now we are ready to prove the total time complexity of \texttt{CF-CTCP} (Lemma~\ref{lemma:time-complexity-CF-CTCP}).
\begin{proof}
    Note that we invoke \texttt{CF-CTCP} only when $Q_v \neq \emptyset$ or $lb\_changed=true$ as in \texttt{CTCP}~\cite{CXS22}.
    First, for the first invocation, Line 6 of Algorithm~\ref{algorithm:CF-CTCP} computes the common neighbors $\Delta(u,v)$ for each edge $(u,v)$, and the time complexity is $$O(\sum_{(u,v) \in E}\min(d_G(u), d_G(v))) = O(m \times \delta(G)),$$ according to Lemma~\ref{lemma:sum-min-deg-for-edges}.
    Second, the total time consumption of core pruning is $O(m)$~\cite{BVZ03m} and the total time cost of Procedure \texttt{RemoveEdge} is also $O(m)$ since we can implement Line 21 for at most $m$ times.
    Third, for all invocations, there are at most $\delta(G)$ times when $lb\_changed = true$ since $k \leq lb = |S^*| \leq \delta(G)+k$ ($S^*$ denoting the largest $k$-plex seen so far), which indicates that we will perform Lines 4-8 at most $\delta(G)$ times. Thus the total time complexity of Lines 1-8 is $O(m \times \delta(G))$. 
    We next consider Lines 9-20. We will pop at most $m$ edges, and for each edge, we need to find all the triangles that it participates in, which can be done in $O(\sum_{(u,v) \in E}\min(d_G(u), d_G(v))) = O(m \times \delta(G))$. Therefore, the total time complexity of all invocations to \texttt{CF-CTCP} is $O(m \times \delta(G))$, which completes our proof.
\end{proof}

\subsection{An Implementation of \texttt{CF-CTCP} with $O(m)$ Memory} 
A direct implementation of \texttt{CF-CTCP} requires storing the common neighbors $\Delta(\cdot, \cdot)$ for all edges, which needs $O(m\times \delta(G))$ memory. In the following, we propose a novel implementation that requires only $O(m)$ memory without changing the time complexity of \texttt{CF-CTCP}. 
In particular, we need three auxiliary arrays $A_1$, $A_2$, and $A_3$, each of length $m$, to store additional information for each edge: 1) array $A_1$ records the number of triangles, 2) array $A_2$ records the timestamp (e.g., system time) when the triangle count is computed in Line 6, and 3) array $A_3$ records the timestamp (e.g., system time) when an edge is removed in Lines 1, 21 and 22.
Based on these three arrays, we correspondingly modify Algorithm~\ref{algorithm:CF-CTCP} as follows. 
First, we only record $|\Delta(u,v)|$ using $A_1$ instead of storing the whole vertex set $\Delta(u,v)$ in Line 6. The correspond triangle count in $A_1$ is decreased by 1 when \texttt{CF-CTCP} modifies $\Delta(\cdot, \cdot)$ in Lines 13 and 18. 
Second, when we traverse all triangles that edge $(u,v)$ belongs to in Line 12, we enumerate such a vertex $w$ that satisfies: 1) both $(u,w)$ and $(v,w)$ are in $E \cup Q_e$, i.e., $(u,v,w)$ forms a triangle; 2) the timestamp of computing the triangle count for edge $(u,w)$ is before the timestamp of removing edge $(u,v)$ using arrays $A_2$ and $A_3$, i.e., when we compute $|\Delta(u,w)|$ in Line 6, edge $(u,v)$ has not yet been removed. 
The modification to Line 17 follows the same fashion as Line 12.
Finally, it is easy to verify the correctness of the above modification of \texttt{CF-CTCP} with $O(m)$ memory usage.

\section{Additional Experimental Results}
We provide additional experimental results for $k=2, 3, 10, 15, 20$.

\subsection{Comparing with State-of-the-art Algorithms}
\smallskip
\noindent
\textbf{Running times on representative graphs.}
We report the running times of all algorithms on 30 representative graphs with $k=$2, 3, 10, 15, and 20 in Tables~\ref{table:comparing-algorithms-k=2}, ~\ref{table:comparing-algorithms-k=3},~\ref{table:comparing-algorithms-k=10},~\ref{table:comparing-algorithms-k=15}, and~\ref{table:comparing-algorithms-k=20}, respectively. 
We observe that \texttt{kPEX} outperforms all baselines by achieving significant speedups on the majority graphs. For example, \texttt{kPEX} runs at least 5 times faster than \texttt{KPLEX} on 23 out of 30 graphs and at least 5 times faster than \texttt{kPlexT} on 20 out of 30 graphs when $k=3$.

\begin{table}[t]
    \centering
    \small
    \caption{Running time in seconds of \texttt{kPEX} and state-of-the-arts on 30 graphs with $k=2$}
    \scalebox{0.86}{
    \begin{tabular}{r|r|rrrr}
    \hline
        \textbf{ID} & \textbf{\texttt{kPEX} (ours)} & \textbf{\texttt{KPLEX}} & \textbf{\texttt{kPlexT}} & \textbf{\texttt{kPlexS}} & \textbf{\texttt{DiseMKP}} \\ \hline
        G1 & \textbf{0.95} & 1.77 & 1.56 & 2.73 & 1.23 \\ 
        G2 & 1982.78 & \textbf{1847.75} & OOT & OOT & OOT \\ 
        G3 & 51.67 & 105.06 & 128.29 & 178.26 & \textbf{21.63} \\ 
        G4 & \textbf{1.35} & 9.46 & 8.02 & 18.98 & 2.08 \\ 
        G5 & 77.13 & \textbf{27.08} & OOT & OOT & OOT \\ 
        G6 & OOT & OOT & OOT & OOT & OOT \\ 
        G7 & OOT & OOT & OOT & OOT & OOT \\ 
        G8 & \textbf{0.33} & 2.60 & 2.47 & 18.91 & 0.39 \\ 
        G9 & \textbf{39.77} & 86.99 & 87.66 & 637.09 & 497.46 \\ 
        G10 & 4.07 & 33.39 & 30.03 & 479.11 & \textbf{3.30} \\ \hline
        G11 & \textbf{20.89} & 127.32 & 121.80 & 1009.49 & OOT \\ 
        G12 & \textbf{0.54} & 1.45 & 0.70 & 1.56 & 13.80 \\ 
        G13 & \textbf{289.18} & 2505.66 & 2797.56 & OOT & 2401.74 \\ 
        G14 & \textbf{3173.74} & OOT & OOT & OOT & OOT \\ 
        G15 & \textbf{142.05} & 1420.48 & 1627.85 & OOT & OOT \\ 
        G16 & \textbf{1.96} & 7.69 & 8.05 & 31.33 & 7.98 \\ 
        G17 & \textbf{6.30} & 39.70 & 51.66 & 220.65 & 101.81 \\ 
        G18 & \textbf{2.28} & 8.16 & 71.92 & 61.06 & OOT \\ 
        G19 & \textbf{8.71} & 11.07 & 28.43 & 9.39 & OOT \\ 
        G20 & \textbf{5.68} & 23.69 & 77.29 & 19.64 & 328.39 \\ 
        G21 & \textbf{13.00} & 118.76 & 123.49 & 942.62 & 337.35 \\ 
        G22 & \textbf{3.36} & 16.27 & 3.79 & 15.74 & 21.88 \\ 
        G23 & 8.66 & 10.12 & \textbf{3.87} & 10.54 & 1035.93 \\ 
        G24 & \textbf{24.66} & 258.56 & 245.02 & 2345.22 & 550.65 \\ 
        G25 & \textbf{13.90} & 79.16 & 77.10 & 277.82 & OOT \\ 
        G26 & 11.59 & 41.78 & \textbf{5.98} & 107.83 & 232.15 \\ 
        G27 & \textbf{3.18} & 19.32 & 4.39 & 17.14 & 26.33 \\ 
        G28 & \textbf{145.64} & 849.77 & OOT & 1347.36 & OOT \\ 
        G29 & \textbf{6.02} & 7.71 & 2563.26 & 535.31 & OOT \\ 
        G30 & \textbf{147.23} & 514.32 & 721.71 & 1492.62 & OOT \\ \hline
    \end{tabular}
    }
    \label{table:comparing-algorithms-k=2}
\end{table}
\begin{table}[t]
    \centering
    \small
    \caption{Running time in seconds of \texttt{kPEX} and state-of-the-arts on 30 graphs with $k=3$}
    \scalebox{0.86}{
    \begin{tabular}{r|r|rrrr}
    \hline
        \textbf{ID} & \textbf{\texttt{kPEX} (ours)} & \textbf{\texttt{KPLEX}} & \textbf{\texttt{kPlexT}} & \textbf{\texttt{kPlexS}} & \textbf{\texttt{DiseMKP}} \\ \hline
        G1 & \textbf{5.38} & 32.14 & 23.32 & 22.39 & 7.30 \\ 
        G2 & OOT & OOT & OOT & OOT & OOT \\ 
        G3 & 60.32 & 1112.70 & 2071.59 & 1461.59 & \textbf{22.36} \\ 
        G4 & \textbf{9.17} & 269.58 & 69.18 & 705.19 & 10.98 \\ 
        G5 & 0.10 & 0.46 & 1.22 & 15.65 & \textbf{0.08} \\ 
        G6 & \textbf{2.62} & 28.08 & 1552.41 & OOT & 49.09 \\ 
        G7 & OOT & OOT & OOT & OOT & OOT \\ 
        G8 & \textbf{0.16} & 73.99 & 2.78 & 322.77 & 1.02 \\ 
        G9 & \textbf{21.19} & 528.72 & 953.50 & 3592.43 & 2197.53 \\ 
        G10 & 15.04 & 2622.38 & 164.85 & OOT & \textbf{10.05} \\ \hline
        G11 & \textbf{6.55} & OOT & 374.45 & OOT & OOT \\ 
        G12 & \textbf{0.41} & 1.46 & 0.61 & 1.54 & 15.69 \\ 
        G13 & \textbf{163.63} & OOT & OOT & OOT & OOT \\ 
        G14 & OOT & OOT & OOT & OOT & OOT \\ 
        G15 & \textbf{46.87} & OOT & 2702.41 & OOT & OOT \\ 
        G16 & \textbf{1.46} & 167.99 & 7.66 & 184.59 & 20.50 \\ 
        G17 & \textbf{2.14} & 987.10 & 35.01 & 926.33 & 245.60 \\ 
        G18 & \textbf{1.73} & 577.38 & 292.77 & 314.95 & OOT \\ 
        G19 & 0.98 & 3.36 & \textbf{0.85} & 1.34 & 2921.83 \\ 
        G20 & \textbf{4.28} & 102.50 & 26.06 & 25.77 & 317.76 \\ 
        G21 & \textbf{8.88} & OOT & 349.83 & OOT & 1712.85 \\ 
        G22 & \textbf{3.23} & 15.34 & 3.63 & 14.41 & 21.18 \\ 
        G23 & 13.17 & 150.23 & \textbf{4.55} & 10.65 & 1761.28 \\ 
        G24 & \textbf{15.21} & OOT & 643.10 & OOT & OOT \\ 
        G25 & \textbf{6.92} & 2877.27 & 181.87 & 1930.20 & OOT \\ 
        G26 & \textbf{7.07} & 368.98 & 10.68 & 251.34 & 50.96 \\ 
        G27 & \textbf{3.01} & 17.59 & 3.73 & 17.10 & 27.71 \\ 
        G28 & \textbf{117.66} & OOT & OOT & 1163.18 & OOT \\ 
        G29 & \textbf{115.24} & 825.59 & OOT & OOT & OOT \\ 
        G30 & \textbf{200.21} & OOT & OOT & OOT & OOT \\ \hline
    \end{tabular}
    }
    \label{table:comparing-algorithms-k=3}
\end{table}
\begin{table}[t]
    \centering
    \small
    \caption{Running time in seconds of \texttt{kPEX} and state-of-the-arts on 30 graphs with $k=10$}
    \scalebox{0.86}{
    \begin{tabular}{r|r|rrrr}
    \hline
        \textbf{ID} & \textbf{\texttt{kPEX} (ours)} & \textbf{\texttt{KPLEX}} & \textbf{\texttt{kPlexT}} & \textbf{\texttt{kPlexS}} & \textbf{\texttt{DiseMKP}} \\ \hline
        G1 & \textbf{462.04} & OOT & OOT & 3142.77 & OOT \\ 
        G2 & \textbf{38.55} & OOT & OOT & OOT & OOT \\ 
        G3 & OOT & OOT & OOT & OOT & OOT \\ 
        G4 & OOT & OOT & OOT & OOT & OOT \\ 
        G5 & \textbf{14.63} & OOT & OOT & OOT & OOT \\ 
        G6 & OOT & OOT & OOT & OOT & OOT \\ 
        G7 & OOT & OOT & OOT & OOT & OOT \\ 
        G8 & OOT & OOT & OOT & OOT & OOT \\ 
        G9 & \textbf{309.46} & OOT & OOT & OOT & OOT \\ 
        G10 & OOT & OOT & OOT & OOT & OOT \\ \hline
        G11 & \textbf{6.00} & OOT & OOT & OOT & OOT \\ 
        G12 & \textbf{0.50} & 1.12 & 0.58 & 1.40 & OOT \\ 
        G13 & \textbf{859.77} & OOT & OOT & OOT & OOT \\ 
        G14 & \textbf{3201.23} & OOT & OOT & OOT & OOT \\ 
        G15 & \textbf{23.07} & OOT & OOT & OOT & OOT \\ 
        G16 & \textbf{2.01} & 23.08 & OOT & 293.97 & OOT \\ 
        G17 & \textbf{2.03} & 314.96 & OOT & OOT & OOT \\ 
        G18 & \textbf{0.51} & 2679.72 & OOT & 1017.07 & OOT \\ 
        G19 & \textbf{0.29} & 3.04 & 5.21 & 2.65 & 1388.96 \\ 
        G20 & \textbf{5.17} & 409.33 & OOT & 31.76 & OOT \\ 
        G21 & \textbf{20.86} & OOT & OOT & OOT & OOT \\ 
        G22 & \textbf{2.07} & 11.16 & 2.87 & 10.38 & 18.09 \\ 
        G23 & 13.47 & 163.83 & \textbf{4.81} & 12.85 & OOT \\ 
        G24 & \textbf{34.49} & OOT & OOT & OOT & OOT \\ 
        G25 & \textbf{2.81} & 34.57 & OOT & 587.10 & OOT \\ 
        G26 & \textbf{1.10} & 3.61 & 3.73 & 3.90 & 4.04 \\ 
        G27 & \textbf{2.47} & 15.41 & 3.37 & 15.06 & OOT \\ 
        G28 & \textbf{263.01} & OOT & OOT & OOT & OOT \\ 
        G29 & \textbf{2.89} & 101.16 & OOT & OOT & 11.28 \\ 
        G30 & OOT & OOT & OOT & OOT & OOT \\ \hline
    \end{tabular}
    }
    \label{table:comparing-algorithms-k=10}
\end{table}
\begin{table}[t]
    \centering
    \small
    \caption{Running time in seconds of \texttt{kPEX} and state-of-the-arts on 30 graphs with $k=15$}
    \scalebox{0.86}{
    \begin{tabular}{r|r|rrrr}
    \hline
        \textbf{ID} & \textbf{\texttt{kPEX} (ours)} & \textbf{\texttt{KPLEX}} & \textbf{\texttt{kPlexT}} & \textbf{\texttt{kPlexS}} & \textbf{\texttt{DiseMKP}} \\ \hline
        G1 & \textbf{7.84} & 44.29 & 35.84 & 17.20 & 418.76 \\ 
        G2 & \textbf{0.06} & 34.53 & 3.58 & 8.27 & 22.95 \\ 
        G3 & OOT & OOT & OOT & OOT & OOT \\ 
        G4 & OOT & OOT & OOT & OOT & OOT \\ 
        G5 & OOT & OOT & OOT & OOT & OOT \\ 
        G6 & OOT & OOT & OOT & OOT & OOT \\ 
        G7 & OOT & OOT & OOT & OOT & OOT \\ 
        G8 & OOT & OOT & OOT & OOT & OOT \\ 
        G9 & \textbf{115.99} & OOT & OOT & OOT & OOT \\ 
        G10 & OOT & OOT & OOT & OOT & OOT \\ \hline
        G11 & \textbf{18.38} & OOT & OOT & OOT & OOT \\ 
        G12 & \textbf{0.47} & 1.08 & 0.57 & 1.07 & OOT \\ 
        G13 & \textbf{1338.78} & OOT & OOT & OOT & OOT \\ 
        G14 & OOT & OOT & OOT & OOT & OOT \\ 
        G15 & \textbf{141.05} & OOT & OOT & OOT & OOT \\ 
        G16 & \textbf{1.78} & 2.02 & OOT & 2.63 & OOT \\ 
        G17 & \textbf{6.66} & 14.37 & OOT & 321.01 & OOT \\ 
        G18 & \textbf{0.12} & 0.33 & 1144.04 & 0.28 & 2621.37 \\ 
        G19 & \textbf{0.39} & 4.43 & 3.07 & 4.23 & OOT \\ 
        G20 & 594.86 & 194.44 & OOT & \textbf{139.13} & OOT \\ 
        G21 & \textbf{122.70} & OOT & OOT & OOT & OOT \\ 
        G22 & \textbf{1.68} & 8.63 & 33.25 & 9.01 & OOT \\ 
        G23 & 13.21 & 132.50 & \textbf{4.81} & 11.67 & OOT \\ 
        G24 & \textbf{288.71} & OOT & OOT & OOT & OOT \\ 
        G25 & \textbf{5.74} & 168.11 & OOT & OOT & OOT \\ 
        G26 & \textbf{1.12} & 3.63 & 3.57 & 3.88 & 4.02 \\ 
        G27 & \textbf{1.79} & 11.92 & OOT & 9.44 & OOT \\ 
        G28 & OOT & OOT & OOT & OOT & OOT \\ 
        G29 & \textbf{2.46} & 4.08 & 23.15 & 3.69 & 36.37 \\ 
        G30 & OOT & OOT & OOT & OOT & OOT \\ \hline
    \end{tabular}
    }
    \label{table:comparing-algorithms-k=15}
\end{table}
\begin{table}[t]
    \centering
    \small
    \caption{Running time in seconds of \texttt{kPEX} and state-of-the-arts on 30 graphs with $k=20$}
    \scalebox{0.86}{
    \begin{tabular}{r|r|rrrr}
    \hline
        \textbf{ID} & \textbf{\texttt{kPEX} (ours)} & \textbf{\texttt{KPLEX}} & \textbf{\texttt{kPlexT}} & \textbf{\texttt{kPlexS}} & \textbf{\texttt{DiseMKP}} \\ \hline
        G1 & \textbf{0.00} & 0.00 & 0.00 & 0.00 & \textbf{0.00} \\ 
        G2 & 0.00 & 0.00 & 0.00 & 0.00 & \textbf{0.00} \\ 
        G3 & OOT & OOT & OOT & OOT & OOT \\ 
        G4 & OOT & OOT & OOT & OOT & OOT \\ 
        G5 & OOT & OOT & OOT & OOT & OOT \\ 
        G6 & OOT & OOT & OOT & OOT & OOT \\ 
        G7 & OOT & OOT & OOT & OOT & OOT \\ 
        G8 & OOT & OOT & OOT & OOT & OOT \\ 
        G9 & \textbf{36.06} & OOT & OOT & OOT & OOT \\ 
        G10 & OOT & OOT & OOT & OOT & OOT \\ \hline
        G11 & \textbf{35.88} & OOT & OOT & OOT & OOT \\ 
        G12 & \textbf{0.30} & 1.08 & 0.53 & 0.99 & OOT \\ 
        G13 & \textbf{1022.03} & OOT & OOT & OOT & OOT \\ 
        G14 & \textbf{3067.04} & OOT & OOT & OOT & OOT \\ 
        G15 & \textbf{409.76} & OOT & OOT & OOT & OOT \\ 
        G16 & \textbf{9.53} & 311.28 & OOT & 2457.82 & OOT \\ 
        G17 & 47.76 & \textbf{41.72} & OOT & 1801.99 & OOT \\ 
        G18 & \textbf{0.19} & 0.83 & OOT & 0.99 & OOT \\ 
        G19 & \textbf{0.44} & 4.50 & 2.61 & 5.01 & OOT \\ 
        G20 & 114.36 & 1398.04 & OOT & \textbf{59.27} & OOT \\ 
        G21 & \textbf{409.39} & OOT & OOT & OOT & OOT \\ 
        G22 & \textbf{1.77} & 7.64 & OOT & 8.22 & OOT \\ 
        G23 & 12.89 & 130.41 & \textbf{5.05} & 28.77 & OOT \\ 
        G24 & \textbf{1045.59} & OOT & OOT & OOT & OOT \\ 
        G25 & \textbf{10.56} & 334.13 & OOT & OOT & OOT \\ 
        G26 & \textbf{1.05} & 3.80 & 3.42 & 3.67 & 2.70 \\ 
        G27 & \textbf{2.10} & 11.52 & OOT & 11.18 & OOT \\ 
        G28 & OOT & OOT & OOT & OOT & OOT \\ 
        G29 & \textbf{2.47} & 2.99 & 2.95 & 3.28 & 12.19 \\ 
        G30 & OOT & OOT & OOT & OOT & OOT \\ \hline
    \end{tabular}
    }
    \label{table:comparing-algorithms-k=20}
\end{table}

\subsection{Effectiveness of Proposed Techniques}
\smallskip
\noindent
\textbf{Effectiveness of \texttt{AltRB}.} 
We compare \texttt{kPEX} with \texttt{kPEX-SeqRB} and report the running times for $k$=2, 3, 10, 15, and 20 in Tables~\ref{table:ablation-k=2},~\ref{table:ablation-k=3},~\ref{table:ablation-k=10},~\ref{table:ablation-k=15}, and~\ref{table:ablation-k=20}, respectively. 
We observe that \texttt{kPEX} performs better than \texttt{kPEX-SeqRB} at most times, and \texttt{AltRB} can bring at least 60$\times$ speedup on G5 when $k$=10.
In addition, we observe that the gap between \texttt{kPEX} and \texttt{kPEX-AltRB} narrows when $k\geq 15$. A possible reason may be that finding a larger heuristic $k$-plex (i.e., \texttt{KPHeuris}) is more important than \texttt{AltRB} for large values of $k$.

\smallskip
\noindent
\textbf{Effectiveness of \texttt{CF-CTCP}.}
We compare \texttt{kPEX} with \texttt{kPEX-CTCP}, and the running times for $k$=2, 3, 10, 15, and 20 are reported in Tables~\ref{table:ablation-k=2},~\ref{table:ablation-k=3},~\ref{table:ablation-k=10},~\ref{table:ablation-k=15}, and ~\ref{table:ablation-k=20}, respectively. First, \texttt{kPEX} and \texttt{kPEX-CTCP} still have similar performance on G1-G10 because the pre-processing techniques take little time on these small synthetic graphs. 
Second, \texttt{kPEX} runs stably at least 50 times faster than \texttt{kPEX-CTCP} on G23 for all tested values of $k$. 
These results show the effectiveness of \texttt{CF-CTCP} on large sparse graphs.

\smallskip
\noindent
\textbf{Effectiveness of \texttt{KPHeuris}.} 
We compare \texttt{kPEX} with its variants \texttt{kPEX-EGo} and \texttt{kPEX-Degen} (note that \texttt{CF-CTCP} is not replaced). The running times for $k$=2, 3, 10, 15, and 20 are shown in Tables~\ref{table:ablation-k=2},~\ref{table:ablation-k=3},~\ref{table:ablation-k=10},~\ref{table:ablation-k=15}, and ~\ref{table:ablation-k=20}, respectively. 
We have the following observations. First, the running time of \texttt{kPEX} is less than that of both variants on the majority of graphs. Then, \texttt{kPEX} runs up to three orders of magnitude faster than both \texttt{kPEX-EGo} and \texttt{kPEX-Degen} on G19 when $k$=20.
This shows that making more effort to finding a larger initial $k$-plex benefits \texttt{kPEX} by narrowing down the search space.

\smallskip
\noindent
\textbf{Effectiveness of \texttt{KPHeuris} and \texttt{CF-CTCP}.}
We also compare the total pre-processing time and the size of the $k$-plex (i.e., $lb$) obtained by different heuristic methods in \texttt{kPEX}, \texttt{kPlexT}, \texttt{kPlexS}, and \texttt{DiseMKP} (note that \texttt{KPLEX} uses the same pre-processing method as \texttt{kPlexS}).
The results for $k$=2, 3, 10, 15, and 20 are reported in Tables~\ref{table:compare-prepro-k=2},~\ref{table:compare-prepro-k=3},~\ref{table:compare-prepro-k=10},~\ref{table:compare-prepro-k=15}, and~\ref{table:compare-prepro-k=20}, respectively.
Note that we exclude the results on synthetic graphs G1-G10 since they have only hundreds of vertices and can be handled within 1 second by all methods.
We have the following observations.
First, \texttt{kPEX} obtains the largest $lb$ (or matches the largest obtained by others) at most time while the pre-processing time remains comparable to other algorithms.
Second, \texttt{KPHeuris} outperforms the other pre-processing algorithms by obtaining a lager $k$-plex while costing much less time on G20 and G22 for all tested values of $k$ . This also verifies the effectiveness of \texttt{CF-CTCP} and \texttt{KPHeuris}.

\begin{table}[!ht]
    \centering
    \small
    \caption{Running time in seconds of \texttt{kPEX} and its variants on 30 graphs with $k=2$}
    \label{table:ablation-k=2}
    \scalebox{0.86}{
        \begin{tabular}{c|c|cccc}
        \hline
            \textbf{ID} & \textbf{\texttt{kPEX}} & \textbf{\texttt{kPEX-SeqRB}} & \textbf{\texttt{kPEX-CTCP}} & \textbf{\texttt{kPEX-EGo}} & \textbf{\texttt{kPEX-Degen}} \\ 
        \hline
        G1 & \textbf{0.95} & 1.23 & 0.95 & 0.95 & 0.95 \\ 
        G2 & \textbf{1982.78} & 2021.27 & 1988.00 & 1992.98 & 1994.67 \\ 
        G3 & 51.67 & 75.39 & \textbf{51.66} & 52.01 & 52.18 \\ 
        G4 & 1.35 & 2.90 & \textbf{1.34} & 1.36 & 1.36 \\ 
        G5 & \textbf{77.13} & 118.45 & 77.47 & 77.41 & 77.25 \\ 
        G6 & OOT & OOT & OOT & OOT & OOT \\ 
        G7 & OOT & OOT & OOT & OOT & OOT \\ 
        G8 & 0.33 & 0.29 & 0.33 & \textbf{0.29} & 0.31 \\ 
        G9 & \textbf{39.77} & 50.32 & 39.81 & 39.97 & 41.21 \\ 
        G10 & 4.07 & 5.16 & 4.05 & \textbf{3.81} & 3.88 \\ \hline
        G11 & \textbf{20.89} & 33.06 & 21.47 & 23.02 & 25.65 \\ 
        G12 & \textbf{0.54} & 0.55 & 2.41 & 0.96 & 1.12 \\ 
        G13 & \textbf{289.18} & 407.48 & 291.83 & 345.93 & 343.31 \\ 
        G14 & \textbf{3173.74} & OOT & 3186.76 & OOT & OOT \\ 
        G15 & \textbf{142.05} & 189.20 & 158.75 & 172.91 & 277.23 \\ 
        G16 & \textbf{1.96} & 2.00 & 6.23 & 3.34 & 3.33 \\ 
        G17 & \textbf{6.30} & 7.06 & 13.16 & 9.55 & 7.25 \\ 
        G18 & \textbf{2.28} & 2.50 & 2.89 & 2.39 & 2.56 \\ 
        G19 & \textbf{8.71} & 8.85 & 33.79 & 16.26 & OOT \\ 
        G20 & 5.68 & \textbf{4.91} & 245.87 & 190.31 & 316.79 \\ 
        G21 & \textbf{13.00} & 22.35 & 13.87 & 20.14 & 19.76 \\ 
        G22 & 3.36 & \textbf{3.20} & 22.43 & 11.40 & 3.43 \\ 
        G23 & 8.66 & \textbf{8.11} & 1418.09 & 469.20 & 8.54 \\ 
        G24 & 24.66 & 41.83 & 26.23 & \textbf{24.27} & 36.39 \\ 
        G25 & 13.90 & 17.56 & 86.75 & 35.87 & \textbf{13.80} \\ 
        G26 & 11.59 & 21.19 & \textbf{11.03} & 12.31 & 11.56 \\ 
        G27 & \textbf{3.18} & 3.31 & 30.24 & 13.21 & 3.58 \\ 
        G28 & 145.64 & \textbf{137.51} & OOT & OOT & OOT \\ 
        G29 & \textbf{6.02} & 7.65 & 6.46 & 10.13 & 10.28 \\ 
        G30 & \textbf{147.23} & 179.48 & 1453.76 & 1437.70 & 251.08 \\ 
        \hline
        \end{tabular}
    }
\end{table}
\begin{table}[!ht]
    \centering
    \small
    \caption{Running time in seconds of \texttt{kPEX} and its variants on 30 graphs with $k=3$}
    \label{table:ablation-k=3}
    \scalebox{0.86}{
        \begin{tabular}{c|c|cccc}
        \hline
            \textbf{ID} & \textbf{\texttt{kPEX}} & \textbf{\texttt{kPEX-SeqRB}} & \textbf{\texttt{kPEX-CTCP}} & \textbf{\texttt{kPEX-EGo}} & \textbf{\texttt{kPEX-Degen}} \\ 
        \hline
        G1 & \textbf{5.38} & 16.99 & 5.40 & 5.41 & 5.42 \\ 
        G2 & OOT & OOT & OOT & OOT & OOT \\ 
        G3 & \textbf{60.32} & 530.53 & 60.48 & 60.65 & 60.88 \\ 
        G4 & \textbf{9.17} & 59.21 & 9.20 & 9.53 & 9.52 \\ 
        G5 & 0.10 & 0.11 & 0.10 & 0.10 & \textbf{0.09} \\ 
        G6 & 2.62 & 21.44 & \textbf{2.62} & 2.85 & 3.18 \\ 
        G7 & OOT & OOT & OOT & OOT & OOT \\ 
        G8 & 0.16 & 0.18 & \textbf{0.15} & 0.17 & 0.17 \\ 
        G9 & \textbf{21.19} & 149.41 & 21.22 & 25.29 & 28.53 \\ 
        G10 & 15.04 & 82.46 & 15.03 & \textbf{14.80} & 15.20 \\ \hline
        G11 & \textbf{6.55} & 33.43 & 7.15 & 11.34 & 12.64 \\ 
        G12 & \textbf{0.41} & 0.41 & 2.29 & 0.92 & 0.91 \\ 
        G13 & \textbf{163.63} & 985.84 & 167.20 & 195.70 & 208.35 \\ 
        G14 & OOT & OOT & OOT & OOT & OOT \\ 
        G15 & \textbf{46.87} & 204.47 & 62.51 & 63.44 & 57.47 \\ 
        G16 & \textbf{1.46} & 1.68 & 5.91 & 2.26 & 1.87 \\ 
        G17 & \textbf{2.14} & 3.31 & 8.94 & 4.53 & 2.76 \\ 
        G18 & \textbf{1.73} & 3.96 & 2.30 & 2.36 & 2.49 \\ 
        G19 & \textbf{0.98} & 1.00 & 24.22 & 8.57 & 49.43 \\ 
        G20 & \textbf{4.28} & 4.40 & 226.14 & 171.87 & 271.25 \\ 
        G21 & \textbf{8.88} & 48.31 & 9.66 & 9.64 & 9.53 \\ 
        G22 & 3.23 & \textbf{2.81} & 20.21 & 9.99 & 3.34 \\ 
        G23 & 13.17 & \textbf{12.77} & 1617.71 & 468.21 & 118.89 \\ 
        G24 & \textbf{15.21} & 80.36 & 16.67 & 16.34 & 16.61 \\ 
        G25 & \textbf{6.92} & 13.36 & 76.74 & 28.93 & 7.86 \\ 
        G26 & 7.07 & 7.14 & \textbf{6.34} & 7.57 & 6.75 \\ 
        G27 & 3.01 & \textbf{2.88} & 27.60 & 12.35 & 3.18 \\ 
        G28 & 117.66 & \textbf{116.05} & OOT & OOT & 3271.82 \\ 
        G29 & 115.24 & 455.48 & 115.38 & \textbf{113.24} & 118.54 \\ 
        G30 & \textbf{200.21} & 833.98 & 1502.80 & 1356.16 & 353.15 \\ 
        \hline
        \end{tabular}
    }
\end{table}
\begin{table}[!ht]
    \centering
    \small
    \caption{Running time in seconds of \texttt{kPEX} and its variants on 30 graphs with $k=10$}
    \label{table:ablation-k=10}
    \scalebox{0.86}{
        \begin{tabular}{c|c|cccc}
        \hline
            \textbf{ID} & \textbf{\texttt{kPEX}} & \textbf{\texttt{kPEX-SeqRB}} & \textbf{\texttt{kPEX-CTCP}} & \textbf{\texttt{kPEX-EGo}} & \textbf{\texttt{kPEX-Degen}} \\ 
        \hline
        G1 & \textbf{462.04} & 1238.20 & 463.16 & 464.53 & 464.13 \\ 
        G2 & \textbf{38.55} & 50.23 & 38.79 & 38.71 & 38.58 \\ 
        G3 & OOT & OOT & OOT & OOT & OOT \\ 
        G4 & OOT & OOT & OOT & OOT & OOT \\ 
        G5 & \textbf{14.63} & 912.36 & 14.65 & 14.72 & 14.67 \\ 
        G6 & OOT & OOT & OOT & OOT & OOT \\ 
        G7 & OOT & OOT & OOT & OOT & OOT \\ 
        G8 & OOT & OOT & OOT & OOT & OOT \\ 
        G9 & \textbf{309.46} & 837.83 & 310.32 & 625.51 & 623.68 \\ 
        G10 & OOT & OOT & OOT & OOT & OOT \\ \hline
        G11 & \textbf{6.00} & 26.39 & 6.66 & 17.61 & 17.30 \\ 
        G12 & \textbf{0.50} & 0.52 & 2.22 & 0.83 & 0.60 \\ 
        G13 & \textbf{859.77} & OOT & 866.51 & 1206.11 & 1200.79 \\ 
        G14 & \textbf{3201.23} & OOT & 3207.34 & 3366.27 & 3346.56 \\ 
        G15 & \textbf{23.07} & 234.55 & 35.02 & 35.29 & 31.29 \\ 
        G16 & \textbf{2.01} & 3.71 & 6.10 & 4.78 & 4.75 \\ 
        G17 & \textbf{2.03} & 13.22 & 7.91 & 6.87 & 5.07 \\ 
        G18 & \textbf{0.51} & 21.17 & 1.09 & 3.77 & 3.57 \\ 
        G19 & \textbf{0.29} & 0.29 & 2.81 & 44.58 & OOT \\ 
        G20 & \textbf{5.17} & 7.01 & 204.72 & 108.26 & 225.67 \\ 
        G21 & \textbf{20.86} & 125.24 & 21.52 & 50.77 & 50.31 \\ 
        G22 & 2.07 & \textbf{1.97} & 15.59 & 7.08 & 2.05 \\ 
        G23 & 13.47 & \textbf{12.81} & 1613.00 & 462.58 & 116.71 \\ 
        G24 & \textbf{34.49} & 195.37 & 35.84 & 52.85 & 51.96 \\ 
        G25 & 2.81 & \textbf{2.73} & 58.16 & 17.78 & 2.91 \\ 
        G26 & 1.10 & \textbf{1.09} & 2.87 & 2.55 & 1.48 \\ 
        G27 & 2.47 & \textbf{2.33} & 23.40 & 9.59 & 2.38 \\ 
        G28 & 263.01 & \textbf{261.12} & OOT & OOT & OOT \\ 
        G29 & 2.89 & 3.18 & 2.86 & 2.76 & \textbf{2.58} \\ 
        G30 & OOT & OOT & OOT & OOT & OOT \\ 
        \hline
        \end{tabular}
    }
\end{table}
\begin{table}[!ht]
    \centering
    \small
    \caption{Running time in seconds of \texttt{kPEX} and its variants on 30 graphs with $k=15$}
    \label{table:ablation-k=15}
    \scalebox{0.86}{
        \begin{tabular}{c|c|cccc}
        \hline
            \textbf{ID} & \textbf{\texttt{kPEX}} & \textbf{\texttt{kPEX-SeqRB}} & \textbf{\texttt{kPEX-CTCP}} & \textbf{\texttt{kPEX-EGo}} & \textbf{\texttt{kPEX-Degen}} \\ 
        \hline
        G1 & 7.84 & 10.63 & 7.85 & \textbf{7.72} & 7.94 \\ 
        G2 & 0.06 & 0.06 & 0.05 & 0.06 & \textbf{0.05} \\ 
        G3 & OOT & OOT & OOT & OOT & OOT \\ 
        G4 & OOT & OOT & OOT & OOT & OOT \\ 
        G5 & OOT & OOT & OOT & OOT & OOT \\ 
        G6 & OOT & OOT & OOT & OOT & OOT \\ 
        G7 & OOT & OOT & OOT & OOT & OOT \\ 
        G8 & OOT & OOT & OOT & OOT & OOT \\ 
        G9 & 115.99 & 120.17 & 115.87 & \textbf{113.30} & 118.43 \\ 
        G10 & OOT & OOT & OOT & OOT & OOT \\ \hline
        G11 & 18.38 & 19.24 & \textbf{18.37} & 71.93 & 74.15 \\ 
        G12 & \textbf{0.47} & 0.47 & 2.11 & 0.78 & 0.62 \\ 
        G13 & 1338.78 & 1497.62 & \textbf{1329.66} & 2796.92 & 2867.16 \\ 
        G14 & OOT & OOT & OOT & OOT & OOT \\ 
        G15 & \textbf{141.05} & 167.43 & 148.50 & 170.58 & 173.08 \\ 
        G16 & \textbf{1.78} & 1.86 & 5.55 & 2.48 & 2.98 \\ 
        G17 & 6.66 & 7.52 & 12.19 & 7.16 & \textbf{5.70} \\ 
        G18 & \textbf{0.12} & 0.13 & 0.67 & 0.34 & 0.18 \\ 
        G19 & \textbf{0.39} & 0.40 & 0.70 & 728.27 & OOT \\ 
        G20 & \textbf{594.86} & 1240.20 & 785.65 & 943.65 & 1486.61 \\ 
        G21 & 122.70 & 131.19 & \textbf{121.98} & 225.87 & 233.15 \\ 
        G22 & 1.68 & 1.68 & 12.05 & 5.75 & \textbf{1.66} \\ 
        G23 & 13.21 & \textbf{13.15} & 1605.10 & 458.65 & 116.44 \\ 
        G24 & 288.71 & 316.40 & 286.53 & \textbf{281.00} & 293.41 \\ 
        G25 & 5.74 & 5.74 & 55.64 & 19.29 & \textbf{5.72} \\ 
        G26 & 1.12 & \textbf{1.10} & 2.72 & 2.49 & 1.48 \\ 
        G27 & 1.79 & \textbf{1.75} & 17.87 & 7.54 & 1.81 \\ 
        G28 & OOT & OOT & OOT & OOT & OOT \\ 
        G29 & 2.46 & 2.57 & 2.54 & 2.15 & \textbf{2.10} \\ 
        G30 & OOT & OOT & OOT & OOT & OOT \\ 
        \hline
        \end{tabular}
    }
\end{table}
\begin{table}[!ht]
    \centering
    \small
    \caption{Running time in seconds of \texttt{kPEX} and its variants on 30 graphs with $k=20$}
    \label{table:ablation-k=20}
    \scalebox{0.86}{
        \begin{tabular}{c|c|cccc}
        \hline
            \textbf{ID} & \textbf{\texttt{kPEX}} & \textbf{\texttt{kPEX-SeqRB}} & \textbf{\texttt{kPEX-CTCP}} & \textbf{\texttt{kPEX-EGo}} & \textbf{\texttt{kPEX-Degen}} \\ 
        \hline
        G1 & \textbf{0.00} & 0.00 & \textbf{0.00} & 0.00 & \textbf{0.00} \\ 
        G2 & 0.00 & 0.00 & 0.00 & \textbf{0.00} & \textbf{0.00} \\ 
        G3 & OOT & OOT & OOT & OOT & OOT \\ 
        G4 & OOT & OOT & OOT & OOT & OOT \\ 
        G5 & OOT & OOT & OOT & OOT & OOT \\ 
        G6 & OOT & OOT & OOT & OOT & OOT \\ 
        G7 & OOT & OOT & OOT & OOT & OOT \\ 
        G8 & OOT & OOT & OOT & OOT & OOT \\ 
        G9 & 36.06 & \textbf{35.42} & 35.87 & 71.18 & 74.99 \\ 
        G10 & OOT & OOT & OOT & OOT & OOT \\ \hline
        G11 & 35.88 & 35.96 & \textbf{35.82} & 58.23 & 59.62 \\ 
        G12 & \textbf{0.30} & 0.30 & 1.88 & 0.81 & 0.56 \\ 
        G13 & 1022.03 & \textbf{1010.12} & 1020.25 & OOT & OOT \\ 
        G14 & 3067.04 & 3450.98 & \textbf{3055.50} & OOT & OOT \\ 
        G15 & \textbf{409.76} & 416.74 & 412.44 & 1627.16 & 1688.80 \\ 
        G16 & 9.53 & 10.78 & 10.71 & \textbf{9.37} & 24.72 \\ 
        G17 & \textbf{47.76} & 58.52 & 52.69 & 48.09 & 49.33 \\ 
        G18 & \textbf{0.19} & 0.19 & 0.67 & 0.35 & 0.21 \\ 
        G19 & \textbf{0.44} & 0.45 & 0.57 & OOT & OOT \\ 
        G20 & \textbf{114.36} & 122.56 & 274.94 & 625.36 & OOT \\ 
        G21 & 409.39 & 408.68 & \textbf{407.72} & 572.13 & 600.41 \\ 
        G22 & 1.77 & \textbf{1.76} & 12.48 & 5.96 & 2.10 \\ 
        G23 & 12.89 & \textbf{12.65} & 1602.86 & 452.09 & 116.28 \\ 
        G24 & 1045.59 & 1060.97 & \textbf{1040.63} & 1134.76 & 1191.44 \\ 
        G25 & 10.56 & \textbf{10.50} & 58.15 & 24.76 & 11.29 \\ 
        G26 & \textbf{1.05} & 1.10 & 2.36 & 2.34 & 1.39 \\ 
        G27 & \textbf{2.10} & 2.10 & 16.94 & 7.54 & 2.17 \\ 
        G28 & OOT & OOT & OOT & OOT & OOT \\ 
        G29 & 2.47 & 2.47 & 2.45 & 1.91 & \textbf{1.84} \\ 
        G30 & OOT & OOT & OOT & OOT & \textbf{878.37} \\ 
        \hline
        \end{tabular}
    }
\end{table}

\begin{table}[!ht]
    \caption{Pre-processing time in seconds on 20 graphs with $k$=2
    ($lb$ denotes the size of the computed heuristic $k$-plex)
    }
    \label{table:compare-prepro-k=2}
    \centering
    \small
    \scalebox{0.99}{
        \begin{tabular}{r|cc|cc|cc|cc}
        \hline
        \multirow{2}{*}{\textbf{ID}} & \multicolumn{2}{c|}{\textbf{\texttt{kPEX}}} & \multicolumn{2}{c|}{\textbf{\texttt{kPlexT}}} & \multicolumn{2}{c|}{\textbf{\texttt{kPlexS}}} & \multicolumn{2}{c}{\textbf{\texttt{DiseMKP}}} \\
         & \textbf{time} & \textbf{$lb$} & \textbf{time} & \textbf{$lb$} & \textbf{time} & \textbf{$lb$} & \textbf{time} & \textbf{$lb$} \\ 
         \hline
         G11 & 1.08 & \textbf{37} & 0.46 & \textbf{37} & \textbf{0.27} & 34 & 0.54 & 35 \\ 
        G12 & \textbf{0.52} & \textbf{63} & 0.77 & 62 & 1.32 & 21 & 1.78 & 23 \\ 
        G13 & 9.08 & \textbf{26} & \textbf{1.64} & 25 & 2.21 & 25 & 3.70 & 23 \\ 
        G14 & 25.60 & \textbf{54} & 5.17 & 52 & \textbf{2.91} & 51 & 6.12 & 52 \\ 
        G15 & 7.72 & \textbf{36} & \textbf{2.80} & \textbf{36} & 4.09 & 32 & 13.11 & 34 \\ 
        G16 & 1.16 & \textbf{19} & \textbf{0.91} & 17 & 1.96 & 9 & 4.07 & 9 \\ 
        G17 & 1.17 & \textbf{28} & \textbf{0.72} & 27 & 1.86 & 27 & 6.02 & 26 \\ 
        G18 & 0.33 & \textbf{70} & 0.42 & \textbf{70} & \textbf{0.29} & 69 & 0.55 & \textbf{70} \\ 
        G19 & 4.48 & \textbf{35} & 1.05 & 34 & \textbf{0.79} & 30 & 15.29 & 31 \\ 
        G20 & \textbf{4.02} & \textbf{15} & 4.10 & 14 & 9.06 & 10 & 141.38 & 4 \\ 
        G21 & 1.01 & \textbf{32} & \textbf{0.41} & 30 & 0.44 & 30 & 0.69 & 28 \\ 
        G22 & \textbf{3.36} & \textbf{31} & 5.58 & \textbf{31} & 20.15 & 17 & 21.32 & 18 \\ 
        G23 & 8.66 & \textbf{6} & \textbf{4.92} & 5 & 13.76 & 5 & 990.37 & 5 \\ 
        G24 & 1.75 & \textbf{32} & \textbf{0.82} & \textbf{32} & 0.98 & 30 & 1.58 & 30 \\ 
        G25 & \textbf{4.37} & \textbf{60} & 8.81 & \textbf{60} & 37.18 & \textbf{60} & 56.61 & 59 \\ 
        G26 & \textbf{3.01} & \textbf{872} & 3.57 & \textbf{872} & 3.76 & \textbf{872} & 3.20 & \textbf{872} \\ 
        G27 & \textbf{3.12} & \textbf{27} & 5.27 & \textbf{27} & 20.55 & 24 & 24.79 & 24 \\ 
        G28 & 102.61 & \textbf{11} & \textbf{72.20} & 10 & 201.25 & 9 & OOT & - \\ 
        G29 & \textbf{2.68} & 273 & 3.28 & \textbf{274} & 3.14 & 271 & 9.38 & 272 \\ 
        G30 & \textbf{108.37} & \textbf{52} & 144.45 & 51 & 287.50 & 10 & 1688.08 & 10 \\ 
        \hline
        \end{tabular}
    }
\end{table}
\begin{table}[!ht]
    \caption{Pre-processing time in seconds on 20 graphs with $k$=3
    ($lb$ denotes the size of the computed heuristic $k$-plex)
    }
    \label{table:compare-prepro-k=3}
    \centering
    \small
    \scalebox{0.99}{
        \begin{tabular}{r|cc|cc|cc|cc}
        \hline
        \multirow{2}{*}{\textbf{ID}} & \multicolumn{2}{c|}{\textbf{\texttt{kPEX}}} & \multicolumn{2}{c|}{\textbf{\texttt{kPlexT}}} & \multicolumn{2}{c|}{\textbf{\texttt{kPlexS}}} & \multicolumn{2}{c}{\textbf{\texttt{DiseMKP}}} \\
         & \textbf{time} & \textbf{$lb$} & \textbf{time} & \textbf{$lb$} & \textbf{time} & \textbf{$lb$} & \textbf{time} & \textbf{$lb$} \\ 
         \hline
         G11 & 0.82 & \textbf{44} & 0.45 & 42 & \textbf{0.28} & 40 & 0.55 & 41 \\ 
        G12 & \textbf{0.40} & \textbf{67} & 0.70 & 66 & 1.38 & 26 & 1.72 & 27 \\ 
        G13 & 7.58 & \textbf{30} & \textbf{1.89} & 29 & 2.28 & 28 & 3.72 & 27 \\ 
        G14 & 11.84 & \textbf{58} & 5.23 & \textbf{58} & \textbf{3.23} & 56 & 6.17 & 56 \\ 
        G15 & 4.70 & \textbf{41} & \textbf{2.24} & 40 & 4.22 & 40 & 11.86 & 39 \\ 
        G16 & 1.10 & \textbf{22} & \textbf{0.78} & 21 & 1.76 & 11 & 4.15 & 10 \\ 
        G17 & 0.84 & \textbf{33} & \textbf{0.70} & 32 & 1.79 & 31 & 5.91 & 29 \\ 
        G18 & 0.45 & \textbf{77} & 0.40 & 75 & \textbf{0.28} & 74 & 0.54 & 74 \\ 
        G19 & 0.97 & \textbf{39} & 0.98 & \textbf{39} & \textbf{0.74} & 34 & 14.34 & 34 \\ 
        G20 & \textbf{3.39} & \textbf{17} & 4.49 & 16 & 8.21 & 10 & 133.31 & 6 \\ 
        G21 & 0.99 & \textbf{35} & \textbf{0.44} & 34 & 0.45 & 33 & 0.68 & 33 \\ 
        G22 & \textbf{3.23} & \textbf{32} & 5.63 & 31 & 20.64 & 20 & 22.30 & 18 \\ 
        G23 & 13.17 & \textbf{7} & \textbf{5.79} & \textbf{7} & 11.91 & 6 & 1723.31 & 6 \\ 
        G24 & 1.70 & \textbf{35} & \textbf{0.84} & 34 & 1.00 & 31 & 1.58 & 31 \\ 
        G25 & \textbf{3.83} & \textbf{66} & 9.69 & \textbf{66} & 35.19 & 64 & 52.06 & 65 \\ 
        G26 & \textbf{3.24} & \textbf{875} & 3.51 & \textbf{875} & 3.57 & \textbf{875} & 3.32 & \textbf{875} \\ 
        G27 & \textbf{3.01} & \textbf{32} & 5.16 & \textbf{32} & 19.15 & 27 & 24.36 & 28 \\ 
        G28 & \textbf{77.67} & \textbf{13} & 152.26 & 12 & 320.38 & 9 & OOT & - \\ 
        G29 & \textbf{2.53} & \textbf{280} & 3.12 & \textbf{280} & 3.47 & \textbf{280} & 9.15 & 279 \\ 
        G30 & \textbf{118.35} & \textbf{57} & 143.21 & 55 & 259.23 & 12 & 1481.51 & 12 \\ 
        \hline
        \end{tabular}
    }
\end{table}
\begin{table}[!ht]
    \caption{Pre-processing time in seconds on 20 graphs with $k$=10
    ($lb$ denotes the size of the computed heuristic $k$-plex)
    }
    \label{table:compare-prepro-k=10}
    \centering
    \small
    \scalebox{0.99}{
        \begin{tabular}{r|cc|cc|cc|cc}
        \hline
        \multirow{2}{*}{\textbf{ID}} & \multicolumn{2}{c|}{\textbf{\texttt{kPEX}}} & \multicolumn{2}{c|}{\textbf{\texttt{kPlexT}}} & \multicolumn{2}{c|}{\textbf{\texttt{kPlexS}}} & \multicolumn{2}{c}{\textbf{\texttt{DiseMKP}}} \\
         & \textbf{time} & \textbf{$lb$} & \textbf{time} & \textbf{$lb$} & \textbf{time} & \textbf{$lb$} & \textbf{time} & \textbf{$lb$} \\ 
         \hline
         G11 & \textbf{0.25} & \textbf{67} & 0.37 & 65 & 0.28 & 65 & 0.49 & \textbf{67} \\ 
        G12 & \textbf{0.50} & \textbf{82} & 0.64 & 74 & 1.18 & 45 & 1.58 & 45 \\ 
        G13 & 2.55 & \textbf{54} & \textbf{1.42} & 52 & 2.55 & 52 & 3.91 & \textbf{54} \\ 
        G14 & 4.58 & \textbf{90} & 4.55 & 88 & \textbf{3.39} & 88 & 6.06 & 87 \\ 
        G15 & \textbf{1.84} & \textbf{65} & 1.85 & 64 & 3.73 & 64 & 8.75 & 62 \\ 
        G16 & \textbf{0.73} & \textbf{40} & 0.77 & 36 & 1.71 & 25 & 3.52 & 24 \\ 
        G17 & 0.64 & \textbf{53} & \textbf{0.59} & 48 & 1.66 & 48 & 5.04 & 49 \\ 
        G18 & \textbf{0.09} & \textbf{102} & 0.29 & 101 & 0.23 & 101 & 0.46 & 101 \\ 
        G19 & \textbf{0.28} & \textbf{53} & 5.25 & 47 & 1.18 & 44 & 1.75 & 44 \\ 
        G20 & \textbf{3.99} & \textbf{31} & 33.64 & 29 & 31.91 & 20 & 79.90 & 20 \\ 
        G21 & 0.53 & \textbf{57} & 0.55 & 54 & \textbf{0.42} & 54 & 0.61 & 54 \\ 
        G22 & \textbf{2.07} & \textbf{44} & 3.94 & 38 & 13.23 & 34 & 13.02 & 35 \\ 
        G23 & 13.47 & \textbf{21} & \textbf{5.90} & \textbf{21} & 13.03 & 20 & 1688.87 & 20 \\ 
        G24 & 0.89 & \textbf{57} & \textbf{0.63} & 55 & 0.95 & 55 & 1.26 & 54 \\ 
        G25 & \textbf{2.43} & \textbf{92} & 6.75 & 91 & 30.65 & 91 & 40.20 & 91 \\ 
        G26 & \textbf{1.10} & \textbf{891} & 4.33 & \textbf{891} & 4.75 & \textbf{891} & 4.20 & \textbf{891} \\ 
        G27 & \textbf{2.47} & \textbf{46} & 4.62 & 45 & 17.20 & 40 & 19.06 & 41 \\ 
        G28 & \textbf{241.44} & \textbf{27} & 1808.29 & 22 & 3247.72 & 21 & OOT & - \\ 
        G29 & \textbf{2.44} & \textbf{316} & 3.25 & \textbf{316} & 3.33 & \textbf{316} & 9.32 & 315 \\ 
        G30 & \textbf{99.21} & \textbf{82} & 124.70 & 77 & 229.23 & 26 & 887.60 & 26 \\ 
        \hline
        \end{tabular}
    }
\end{table}
\begin{table}[!ht]
    \caption{Pre-processing time in seconds on 20 graphs with $k$=15
    ($lb$ denotes the size of the computed heuristic $k$-plex)
    }
    \label{table:compare-prepro-k=15}
    \centering
    \small
    \scalebox{0.99}{
        \begin{tabular}{r|cc|cc|cc|cc}
        \hline
        \multirow{2}{*}{\textbf{ID}} & \multicolumn{2}{c|}{\textbf{\texttt{kPEX}}} & \multicolumn{2}{c|}{\textbf{\texttt{kPlexT}}} & \multicolumn{2}{c|}{\textbf{\texttt{kPlexS}}} & \multicolumn{2}{c}{\textbf{\texttt{DiseMKP}}} \\
         & \textbf{time} & \textbf{$lb$} & \textbf{time} & \textbf{$lb$} & \textbf{time} & \textbf{$lb$} & \textbf{time} & \textbf{$lb$} \\ 
         \hline
         G11 & \textbf{0.20} & \textbf{79} & 0.34 & 77 & 0.26 & 77 & 0.45 & 77 \\ 
        G12 & \textbf{0.46} & \textbf{89} & 0.62 & 78 & 1.16 & 54 & 1.47 & 55 \\ 
        G13 & 2.02 & \textbf{67} & \textbf{1.37} & 63 & 2.42 & 63 & 3.86 & 65 \\ 
        G14 & \textbf{2.85} & \textbf{108} & 4.38 & 107 & 3.41 & 107 & 5.90 & 107 \\ 
        G15 & 1.46 & \textbf{77} & \textbf{1.39} & 76 & 3.45 & 76 & 7.58 & 74 \\ 
        G16 & 0.86 & \textbf{51} & \textbf{0.83} & 47 & 1.55 & 33 & 3.01 & 33 \\ 
        G17 & \textbf{0.45} & \textbf{64} & 0.56 & 59 & 1.57 & 59 & 4.76 & 57 \\ 
        G18 & \textbf{0.06} & \textbf{116} & 0.25 & 115 & 0.21 & 115 & 0.41 & 115 \\ 
        G19 & \textbf{0.38} & \textbf{59} & 2.35 & 50 & 1.22 & 49 & 0.42 & 49 \\ 
        G20 & \textbf{4.80} & \textbf{41} & 28.86 & 37 & 28.31 & 30 & 59.46 & 30 \\ 
        G21 & 0.44 & \textbf{69} & 0.52 & 67 & \textbf{0.40} & 67 & 0.59 & 68 \\ 
        G22 & \textbf{1.68} & \textbf{49} & 3.49 & 42 & 10.83 & 42 & 10.67 & 42 \\ 
        G23 & 13.21 & \textbf{31} & \textbf{5.81} & \textbf{31} & 12.79 & 30 & 1713.30 & 30 \\ 
        G24 & 0.66 & \textbf{69} & \textbf{0.59} & \textbf{69} & 0.82 & \textbf{69} & 1.15 & 66 \\ 
        G25 & \textbf{2.26} & \textbf{101} & 6.66 & \textbf{101} & 29.13 & \textbf{101} & 37.67 & \textbf{101} \\ 
        G26 & \textbf{1.12} & \textbf{900} & 4.22 & \textbf{900} & 4.24 & \textbf{900} & 3.91 & \textbf{900} \\ 
        G27 & \textbf{1.78} & \textbf{53} & 3.71 & 51 & 13.72 & 51 & 17.02 & 50 \\ 
        G28 & \textbf{469.30} & \textbf{36} & 1307.77 & 31 & 3177.53 & 31 & OOT & - \\ 
        G29 & \textbf{2.32} & \textbf{332} & 3.62 & \textbf{332} & 3.72 & \textbf{332} & 9.04 & 331 \\ 
        G30 & \textbf{86.23} & \textbf{98} & 111.59 & 92 & 215.72 & 36 & 659.92 & 36 \\ 
        \hline
        \end{tabular}
    }
\end{table}
\begin{table}[!ht]
    \caption{Pre-processing time in seconds on 20 graphs with $k$=20
    ($lb$ denotes the size of the computed heuristic $k$-plex)
    }
    \label{table:compare-prepro-k=20}
    \centering
    \small
    \scalebox{0.99}{
        \begin{tabular}{r|cc|cc|cc|cc}
        \hline
        \multirow{2}{*}{\textbf{ID}} & \multicolumn{2}{c|}{\textbf{\texttt{kPEX}}} & \multicolumn{2}{c|}{\textbf{\texttt{kPlexT}}} & \multicolumn{2}{c|}{\textbf{\texttt{kPlexS}}} & \multicolumn{2}{c}{\textbf{\texttt{DiseMKP}}} \\
         & \textbf{time} & \textbf{$lb$} & \textbf{time} & \textbf{$lb$} & \textbf{time} & \textbf{$lb$} & \textbf{time} & \textbf{$lb$} \\ 
         \hline
         G11 & \textbf{0.17} & \textbf{89} & 0.31 & 88 & 0.24 & 88 & 0.41 & 87 \\ 
        G12 & \textbf{0.30} & \textbf{96} & 0.58 & 78 & 1.08 & 64 & 1.42 & 63 \\ 
        G13 & 2.36 & \textbf{79} & \textbf{1.33} & 76 & 2.40 & 76 & 3.60 & 77 \\ 
        G14 & \textbf{2.18} & \textbf{123} & 4.21 & 122 & 3.25 & 122 & 5.70 & 121 \\ 
        G15 & \textbf{1.33} & \textbf{88} & 1.67 & 85 & 3.35 & 85 & 6.55 & 86 \\ 
        G16 & 2.75 & \textbf{59} & 1.31 & 51 & \textbf{0.80} & 40 & 1.41 & 40 \\ 
        G17 & \textbf{0.46} & \textbf{74} & 0.57 & 67 & 1.51 & 67 & 4.54 & 67 \\ 
        G18 & \textbf{0.06} & \textbf{124} & 0.24 & 123 & 0.20 & 123 & 0.39 & 123 \\ 
        G19 & 0.33 & \textbf{64} & 1.08 & 54 & 1.01 & 54 & \textbf{0.31} & 54 \\ 
        G20 & \textbf{5.30} & \textbf{50} & 31.26 & 44 & 25.04 & 40 & 45.51 & 40 \\ 
        G21 & 0.64 & \textbf{79} & 0.49 & 76 & \textbf{0.39} & 76 & 0.57 & 78 \\ 
        G22 & \textbf{1.76} & \textbf{55} & 5.07 & 47 & 11.81 & 47 & 11.71 & 48 \\ 
        G23 & 12.89 & \textbf{41} & \textbf{5.90} & \textbf{41} & 13.27 & 40 & 1695.05 & 40 \\ 
        G24 & 0.77 & \textbf{78} & \textbf{0.58} & 77 & 0.76 & 77 & 1.08 & 75 \\ 
        G25 & \textbf{2.22} & \textbf{111} & 6.75 & 110 & 25.55 & 110 & 35.31 & 110 \\ 
        G26 & \textbf{1.05} & \textbf{910} & 4.25 & \textbf{910} & 4.29 & \textbf{910} & 2.94 & \textbf{910} \\ 
        G27 & \textbf{1.76} & \textbf{60} & 3.96 & 58 & 12.15 & 58 & 14.96 & 58 \\ 
        G28 & 1986.00 & \textbf{45} & \textbf{779.28} & 40 & 2487.34 & 40 & 947.97 & 40 \\ 
        G29 & \textbf{2.38} & \textbf{343} & 3.33 & \textbf{343} & 3.62 & \textbf{343} & 8.70 & 342 \\ 
        G30 & \textbf{76.54} & \textbf{112} & 98.14 & 104 & 200.47 & 46 & 554.62 & 45 \\ 
        \hline
        \end{tabular}
    }
\end{table}

\clearpage

\end{document}